\documentclass[11pt, onecolumn, draftcls, a4paper]{IEEEtran}

\usepackage{epsfig}
\usepackage{subfigure}
\usepackage{amsmath}
\usepackage{float}
\usepackage{graphicx}
\usepackage{latexsym}
\usepackage{amssymb}
\usepackage{amsfonts}
\usepackage{cite}
\usepackage{url}

\usepackage{tikz}
\usepackage{pgfplots}
\usepackage{verbatim}
\usetikzlibrary{arrows}
\usetikzlibrary{positioning}
\usetikzlibrary{plotmarks}
\usetikzlibrary{decorations}
\usetikzlibrary{shapes}
\usetikzlibrary{plotmarks}

\usepackage{etex}
\usepackage[ngerman]{babel}
\usepackage[T1]{fontenc}
\usepackage{pgfplots}
\usepackage{flushend}
\usepackage{algorithm}
\usepackage{algpseudocode}

\DeclareMathOperator{\ba}{\mathbf{a}}
\DeclareMathOperator{\bb}{\mathbf{b}}
\DeclareMathOperator{\bc}{\mathbf{c}}

\DeclareMathOperator{\e}{\mathbf{e}}
\DeclareMathOperator{\g}{\mathbf{g}}
\DeclareMathOperator{\h}{\mathbf{h}}
\DeclareMathOperator{\n}{\mathbf{n}}
\DeclareMathOperator{\p}{\mathbf{p}}
\DeclareMathOperator{\q}{\mathbf{q}}
\DeclareMathOperator{\s}{\mathbf{s}}
\DeclareMathOperator{\bu}{\mathbf{u}}
\DeclareMathOperator{\bv}{\mathbf{v}}
\DeclareMathOperator{\w}{\mathbf{w}}
\DeclareMathOperator{\x}{\mathbf{x}}
\DeclareMathOperator{\y}{\mathbf{y}}
\DeclareMathOperator{\A}{\mathbf{A}}
\DeclareMathOperator{\B}{\mathbf{B}}
\DeclareMathOperator{\C}{\mathbf{C}}
\DeclareMathOperator{\D}{\mathbf{D}}

\DeclareMathOperator{\F}{\mathbf{F}}
\DeclareMathOperator{\G}{\mathbf{G}}
\DeclareMathOperator{\bH}{\mathbf{H}}
\DeclareMathOperator{\I}{\mathbf{I}}

\DeclareMathOperator{\bL}{\mathbf{L}}
\DeclareMathOperator{\Q}{\mathbf{Q}}
\DeclareMathOperator{\R}{\mathbf{R}}
\DeclareMathOperator{\bS}{\mathbf{S}}
\DeclareMathOperator{\bbS}{\mathbb{S}}
\DeclareMathOperator{\T}{\mathbf{T}}
\DeclareMathOperator{\V}{\mathbf{V}}
\DeclareMathOperator{\bbV}{\mathbb{V}}

\DeclareMathOperator{\X}{\mathbf{X}}
\DeclareMathOperator{\Y}{\mathbf{Y}}
\DeclareMathOperator{\Z}{\mathbf{Z}}

\DeclareMathOperator{\1}{\mathbf{1}}
\DeclareMathOperator{\0}{\mathbf{0}}

\DeclareMathOperator{\tr}{tr}

\DeclareMathOperator{\bvec}{vec}
\DeclareMathOperator{\diag}{diag}
\DeclareMathOperator{\rank}{rank}
\DeclareMathOperator{\sinr}{SINR}

\DeclareMathOperator{\PB}{PB}

\DeclareMathOperator{\IN}{IN}

\DeclareMathOperator{\veta}{\boldsymbol \eta}
\DeclareMathOperator{\vGamma}{\boldsymbol \Gamma}

\newtheorem{Lemma}{Lemma}
\newtheorem{Theorem}{Theorem}
\newtheorem{Def}{Definition}
\newtheorem{Corollary}{Corollary}
\newtheorem{Prob}{Problem Statement}

\title{Instantaneous Relaying: Optimal Strategies and Interference Neutralization}
\author{Zuleita Ho and Eduard Jorswieck\\
Institut of Communication Technology,\\
Faculty of Electrical and Computer Engineering,\\
Technische Universit\"{a}t Dresden\\
\{zuleita.ho, eduard.jorswieck\}@tu-dresden.de \thanks{This work has been performed in the framework of the European research
project SAPHYRE, which is partly funded by the European Union under 
its FP7 ICT Objective 1.1 - The Network of the Future.} }

\begin{document}
 \maketitle
\begin{abstract}
 \textcolor{black}{In a multi-user wireless network equipped with multiple relay nodes, some relays are more intelligent than other relay nodes. The intelligent relays are able to gather channel state information, perform linear processing and forward signals whereas the dumb relays is only able to serve as amplifiers. As the dumb relays are oblivious to the source and destination nodes, the wireless network can be modeled as a relay network with \emph{smart instantaneous relay} only: the signals of source-destination arrive at the same time as source-relay-destination. Recently, instantaneous relaying is shown to improve the degrees-of-freedom of the network as compared to classical cut-set bound. In this paper, we study an achievable rate region and its boundary of the instantaneous interference relay channel in the scenario of (a) uninformed non-cooperative source-destination nodes (source and destination nodes are not aware of the existence of the relay and are non-cooperative) and (b) informed and cooperative source-destination nodes. Further, we examine the performance of interference neutralization: a relay strategy which is able to cancel interference signals at each destination node in the air. We observe that interference neutralization, although promise to achieve desired degrees-of-freedom, may not be feasible if relay has limited power.
Simulation results show that the optimal relay strategies improve the achievable rate region and provide better user-fairness in both uninformed non-cooperative and informed cooperative  scenarios\footnote{Part of this work is submitted to ISIT 2012 \cite{Ho2011b}.}.}
\end{abstract}


\section{Introduction}
In a wireless network, where the source (S) nodes and destination (D) nodes communicate in a shared medium, e.g. time, frequency, code, the signals interfere with each other in the air and may not be possible to be separated at the destination nodes which then leads to a rate degradation. Interference management techniques are thereby crucial to the ever-increasing demand of quality of service. \textcolor{black}{Advanced interference management techniques, e.g., superposition coding and multi-user decoding, induce difficulty in practical implementation as \emph{all codebooks} and \emph{interference decoding capabilities} are required respectively. We assume in this work that the S and D nodes are not able to perform multi-user encoding and decoding techniques\footnote{For the results on instantaneous relay capacity with multi-user decoding destination nodes, please refer to \cite{Chang2011}.}. To improve the rate performance of the system, one can introduce relays to the system and therefore obtain an interference relay channel (IRC). The relay is responsible for signal boosting and interference managing. We assume that the relay employs an amplify-and-forward (AF) strategy which provides flexibility in implementation as the relay is oblivious to the modulation and coding schemes in the communication between S and D nodes \cite{Berger2009}.}

\textcolor{black}{
The novel notion of \emph{relay-without-delays}, also known as instantaneous relays, originally proposed in \cite{ElGamal2005}, refers to a type of relays that can forward signals consisting of both the current symbol and the symbols in the past, instead of only the past symbols as in 
conventional relays\footnote{In the terminology of half-duplex and full-duplex, instantaneous relay is a kind of full-duplex relays.}. The authors further point out that a relay-without-delays is only a conventional relay with one extra delay on the link from S to D. A visual example is given later in \cite{Lee2011} where in a network with multiple S-D pairs and conventional relay nodes, some \emph{naive} relays only AF signals and some \emph{smart} relays can gather channel state information (CSI) and forward linear processing of signals. As the naive relays do not affect the optimization of the system and are oblivious to S and D, the links from S to D through
the \emph{naive} relays are equivalent to direct links from S to D with one extra delay. Hence, the network is equivalent to an IRC equipped with smart relays-without-delays only. It is shown in \cite{ElGamal2007,Cadambe2009a} that instantaneous IRC can achieve degrees-of-freedom (DOF), or the multiplexing-gain, higher than the classical cut-set bound but can greatly simplify the interference alignment scheme \cite{Gomadam2011}. }

\textcolor{black}{
Interference neutralization (IN) is a technique of canceling, zero-forcing or \emph{neutralizing} interference signals by a careful selection of forwarding strategies when the signal travel through relay nodes before reaching the destination. This general idea has been applied in deterministic channels \cite{Mohajer2008, Mohajer2009a} and in two-hop relay channels \cite{Berger2005, Rankov2007, Berger2009}, also known as multi-user zero-forcing and orthogonalize-and-forward. To distinguish from IN, aligned interference neutralization (AIN), which is a combination of interference alignment and interference neutralization techniques \cite{Gou2011, Lee2011}, is first proposed in a conventional full-duplex two-S, two-relay and two-D network (the 2x2x2 channel) \cite{Gou2011} and is later extended to the IRC with an \emph{instantaneous relay} \cite{Lee2011}. Agreeing with \cite{ElGamal2007}, the achievable DOF with an instantaneous relay is higher than the cut-set bound, achieved with the utilization of AIN. For example, the DOF of a two-user IRC with one instantaneous relay is shown to be 1.5 as compared to 1 in a two-user interference channel (IC) without any relay.}

\textcolor{black}{
Although instantaneous relaying can increase the DOF of the underlying interference channel, the DOF analysis describes only the high SNR performance of the system, leaving the system behavior in low to medium SNR regime unmentioned. For any SNR regimes, the boundary of an achievable rate region, the so-called Pareto boundary, holds particular importance because it is highly desirable in the system's interest to operate on the Pareto boundary, i.e., Pareto efficient: no user can improve its own rate without decreasing other users rate \cite{Larsson2009,Zhang2010a, Ho2010a}. It is the purpose of this paper to characterize the Pareto boundary of the instantaneous IRC with (a) the optimal AF relay processing matrix and (b) interference neutralization. Our work differs from the literature in the following areas. We employ IN in an instantaneous relay channel
 with direct links between S and D, whereas \cite{Berger2005, Rankov2007, Berger2009} consider relays-with-delays and with no direct links between S and D.
To maintain simple design of S and D, we consider no symbol extensions and utilize a memory-less instantaneous relay and IN, whereas AIN (a combination of interference alignment and IN by using instantaneous relay with memory and symbol extensions) is employed in \cite{Gou2011, Lee2011}. Note that symbol extensions are shown to achieve desirable DOF. For details, please refer to \cite{Jafar2007a}. We assume linear relay processing which is considered in all aforementioned work. For non-linear instantaneous relaying, please refer to \cite{Khormuji2010}.}

However, the performance analysis of the IRC is not limited to IN. Other interference management techniques, such as interference alignment and interference forwarding, are shown to achieve promising results. In \cite{Nourani2010}, on a quasi-static channel, interference alignment is implemented on a $M$-user IRC with each node equipped with a single antenna. The power scaling of the relay is computed in order to guarantee a degrees-of-freedom $\frac{M}{2}$ in the channel. In \cite{Maric2008,Yu2009,Zhou2010}, instead of neutralizing interference signals at the destination nodes, the relay is responsible for \emph{forwarding} interference signals to the destination nodes so that the strength of the interference signal is increased and is able to be decoded and subtracted. However, a unified algorithm which is capable of adapting different relay strategies, such as interference alignment, IN and interference forwarding, depending on the channel qualities, is yet an open problem. For the reference of the readers, \textcolor{black}{we summarize the achievable DOF of IRC with instantaneous relays and relays-with-delays (conventional relays) in Table \ref{ta}}\footnote{The notations \# S-D stand for number of source and destination nodes in the system. $N_{s,d}$ is the number of antennas at source and destination nodes. $N_r$ is the number of antenna at the relay. \# R is the number of relay nodes in the system. DOF is the degrees-of-freedom achievable. \textcolor{black}{All relays are full-duplex. IR represents instantaneous relays and CR represents conventional relays.} Methods used to achieve the DOF can be aligned interference neutralization (AIN), interference alignment (IA). ME stands for Markov encoding and MUD stands for multi-user decoding whereas SUD stands for single-user decoding. A potent relay is a relay with unlimited power whereas a non-potent relay subjects to transmit power constraint. The notation (1,2) refers to one receive antenna and 2 transmit antennas at the relay.}. Note that the degrees-of-freedom is only a performance measure in the high SNR regime and is not suitable to be used as a performance metric in the low and medium SNR regime.

\begin{table}
\begin{center}
\textcolor{black}{
\caption{The summary of achievable DOF in different topologies of the interference relay channel. \label{ta}}
\begin{tabular}{|c|c|c|c|c|c|c|c|}
 \hline
 & \# S-D & $N_{s,d}$ & $N_{r}$ & \# R & DOF & IR/CR & Method \\
\hline
& 2 & M & M & 1 & $3M/4$ & IR & AIN \cite{Lee2011}, potent relay\\
\cline{2-8}
SUD S-D & 2 & 1 & 1 & 2 & 2 & CR & AIN \cite{Gou2011}, potent relay \\
\cline{2-8}
& K & 1 & K & 1 & $K$ & IR & IN, non-potent relay (this paper)\\
\cline{2-8}
&K & 1 & K & 1 & $K$ & CR & IN, ME, non-potent relay ($P_r \sim O(P_s^2)$) \cite{Tannious2008}\\
\hline
MUD S-D &2 & 1 & (1,2) & 1 & $2$ & IR & MUD, non-potent relay \cite{Chang2011}\\
\hline
\end{tabular}
}
\end{center}
\vspace*{-2em}
\end{table}

\textcolor{black}{
In this paper we study the performance, in terms of achievable rate regions and rate fairness, of a $K$ users instantaneous IRC. We assume a power constraint at the relay and \emph{linear processing capabilities} on the relay and S-D pairs. Global CSI is assumed at the relay which is a common assumption in aforementioned works. Although the effort of CSI estimation and processing at the relays may seem large, the resulted gain largely outweighs the cost in low mobility environments \cite{Berger2009}. A distributed smart relays system, which is capable of gathering CSI, matching the LTE standard, is recently demonstrated \cite{Ylmaz2010}. In Section \ref{sec:ch_model}, we describe the channel model. An achievable rate region and the Pareto boundary formulation for (a) AF strategy optimization and (b) AF and IN strategy are given in Section \ref{sec:pb_formulation}. In the same section, we address the necessary and sufficient condition of feasibility of IN in Theorem \ref{lem:fea}, in terms of channel conditions and the physical attributes of the relay (e.g. the number of antennas and the power at the relay).}

\textcolor{black}{
An improvement of rate performance is expected, as turning off the relay gives the same rate region as the single-input-single-output (SISO) IC. The interesting idea is to improve the rate performance by optimizing the relay strategy, \emph{even if the S-D nodes are uninformed of the relay's presence and have their strategies unaltered (See Section \ref{sec:opt_relay})}. 
The computation of the optimal relay processing matrix, optimal in the sense of Pareto optimality (see Section \ref{sec:pb_formulation} for details), is not trivial as it closely resembles an IC which is well known for the NP-hard complexity of the transmit optimization \cite{Liu2011}, in terms of the number of users in the system. However, we are able to compute an upper bound using the Semi-Definite-Relaxation (SDR) techniques which relaxes the optimization problem to a convex problem and can be solved efficiently. As we will see later in Section \ref{sec:opt_relay}, in the scenario with two S and two D nodes, this relaxation is shown to be tight. 
In Section \ref{sec:opt_power}, we investigate the scenario when the S-D nodes are aware of the relay's presence and are willing to cooperate by optimizing their transmit power.  The resulting rate region can be computed using SDR techniques and can be compared fairly to the achievable rate region of a SISO IC, without any relay, in which the S nodes optimize their transmit power. We show numerical evidence and discussion of the advantages of relay in a system of uninformed S-D nodes in Section \ref{sec:simulations}. Conclusions and future directions are given in Section \ref{sec:conclusion}. }

\subsection{Notations}
The set $\mathbb{C}^{a \times b}$ denotes a set of complex matrices of size $a$ by $b$ and is shortened to $\mathbb{C}^a$ when $a=b$. The notation $\mathcal{N}(\mathbf{A})$ is the null space of $\A$. The set $\mathbb{R}$ denotes all real numbers whereas $\mathbb{R}_+$ denotes the set of positive numbers. The operator $\mathbb{E}_a$ represents expectation over the random variable $a$. The operator $\otimes$ denotes the Kronecker product. The superscripts $^T$, $^H$, $^{-T}$, $^{-H},^\dagger$ represent transpose, Hermitian transpose, inverse and transpose, inverse and Hermitian transpose and Moore-Penrose inverse respectively whereas the superscript $^*$ denotes the conjugation operation. The scalar and vector Euclidean norms are written as $|.|$ and $\|.\|$. The trace of matrix $\mathbf{A}$ is denoted as $\tr(\mathbf{A})$. Vectorization stacks the columns of a matrix $\A$ to form a long column vector denoted as $\bvec(\mathbf{A})$. The identity and zero matrices of dimension $K$ are written as $\I_K$ and $\0_K$. The vector $\e_i$ represents a column vector with zero elements everywhere and one at the $i$-th position. The notation $[\A]_{ml}$ denotes the $m$-th row and $l$-th column element of the matrix $\A$. The notation $\p_{a:b}$, $0 \leq a\leq b\leq n$, denotes a vector which has elements $[p_a, p_{a+1},\ldots, p_b]$ where $\p=[p_1,\ldots,p_n]$.

%

\section{Channel Model}\label{sec:ch_model}

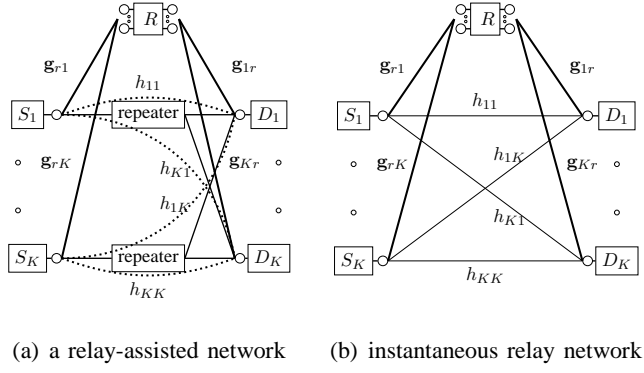
\begin{figure}
\begin{center}
\subfigure[a relay-assisted network]
{
\resizebox{!}{4cm}{
\begin{tikzpicture}
\node[rectangle,draw] (sk) at (0,0) {$S_K$};
\node[rectangle,draw] (s1) at (0,3) {$S_1$};
\node[circle,draw,inner sep=0, text width =0.1cm]  at (-0.2,2) {};
\node[circle,draw,inner sep=0, text width =0.1cm]  at (-0.2,1) {};

\node[circle,draw,inner sep=0, text width =0.2cm] (s1an) at (0.6,3) {};
\node[circle,draw,inner sep=0, text width =0.2cm] (skan) at (0.6,0) {};
\draw[thick] (s1) to (s1an);
\draw[thick] (sk) to (skan);

\node[rectangle,draw] (r2) at (2.5,0) {repeater};
\node[rectangle,draw] (r1) at (2.5,3) {repeater};

\draw[thick] (s1an.east) to (r1.west);
\draw[thick] (skan.east) to (r2.west);

\node[rectangle,draw] (dk) at (5,0) {$D_K$};
\node[rectangle,draw] (d1) at (5,3) {$D_1$};
\node[circle,draw,inner sep=0, text width =0.1cm]  at (5.2,2) {};
\node[circle,draw,inner sep=0, text width =0.1cm]  at (5.2,1) {};

\node[circle,draw,inner sep=0, text width =0.2cm] (d1an) at (4.4,3) {};
\node[circle,draw,inner sep=0, text width =0.2cm] (dkan) at (4.4,0) {};
\draw[thick] (d1) to (d1an);
\draw[thick] (dk) to (dkan);

\draw[thick] (r1.east) to (d1an.west);
\draw[thick] (r1.east) to (dkan.west);
\draw[thick] (r2.east) to (d1an.west);
\draw[thick] (r2.east) to (dkan.west);

\node[rectangle,draw, inner sep=0.2cm, text width =0.2cm] (r) at (2.5,5) {$R$};

\node[circle,draw,inner sep=0,text width =0.2cm] (ran1) at (2,5.2) {};
\node[circle,draw,inner sep=0,text width =0.2cm] (rank) at (2,4.8) {};
\draw[thick] (2.2, 5.2) -- (2.1, 5.2);
\draw[thick] (2.2, 4.8) -- (2.1, 4.8);
\node[circle,draw,inner sep=0,text width =0.05cm]  at (2.1, 5.05) {};
\node[circle,draw,inner sep=0,text width =0.05cm]  at (2.1, 4.95) {};
\node (r_an_l) at (2,5) {};

\node[circle,draw,inner sep=0,text width =0.2cm] (ran1) at (3,5.2) {};
\node[circle,draw,inner sep=0,text width =0.2cm] (rank) at (3,4.8) {};
\draw[thick] (2.8, 5.2) -- (2.9, 5.2);
\draw[thick] (2.8, 4.8) -- (2.9, 4.8);
\node[circle,draw,inner sep=0,text width =0.05cm]  at (2.9, 5.05) {};
\node[circle,draw,inner sep=0,text width =0.05cm]  at (2.9, 4.95) {};
\node (r_an_r) at (3,5) {};

\draw[very thick] (s1an.east) to (r_an_l.west);
\draw[very thick] (skan.east) to (r_an_l.west);
\draw[very thick] (r_an_r.east) to (d1an.west);
\draw[very thick] (r_an_r.east) to (dkan.west);

\node[very thick] at (4.5,2) {$\mathbf{g}_{Kr}$};
\node[very thick] at (0.6,2) {$\mathbf{g}_{rK}$};
\node[very thick] at (4.5,4) {$\mathbf{g}_{1r}$};
\node[very thick] at (0.6,4) {$\mathbf{g}_{r1}$};

\path[very thick, dotted] (s1an.east) edge[bend left=20] node[anchor=south, above] {$h_{11}$} (d1an.west);
\path[very thick, dotted] (s1an.east) edge[bend left=40] node[anchor=north, below] {$h_{K1}$} (dkan.west);
\path[very thick, dotted] (skan.east) edge[bend right=20] node[anchor=north, below] {$h_{KK}$} (dkan.west);
\path[very thick, dotted] (skan.east) edge[bend right=40] node[anchor=south, above] {$h_{1K}$} (d1an.west);

\end{tikzpicture}
}
}
%
%
\subfigure[instantaneous relay network]
{
\resizebox{!}{4cm}{
\begin{tikzpicture}

\node[rectangle,draw] (x1) at (-0.2,3) {$S_1$};
\node[circle,draw,inner sep=0, text width =0.1cm]  at (-0.2,2) {};
\node[circle,draw,inner sep=0, text width =0.1cm]  at (-0.2,1) {};
\node[rectangle,draw] (x2) at (-0.2,0) {$S_K$};

\node[circle,draw,inner sep=0, text width =0.2cm] (x1an) at (0.4,3) {};
\node[circle,draw,inner sep=0, text width =0.2cm] (x2an) at (0.4,0) {};
\draw[thick] (x1) to (x1an);
\draw[thick] (x2) to (x2an);

\node[rectangle,draw] (y1) at (5.2,3) {$D_1$};
\node[rectangle,draw] (y2) at (5.2,0) {$D_K$};
\node[circle,draw,inner sep=0, text width =0.1cm]  at (5.2,2) {};
\node[circle,draw,inner sep=0, text width =0.1cm]  at (5.2,1) {};

\node[circle,draw,inner sep=0, text width =0.2cm] (y1an) at (4.6,3) {};
\node[circle,draw,inner sep=0, text width =0.2cm] (y2an) at (4.6,0) {};
\draw[thick] (y1) to (y1an);
\draw[thick] (y2) to (y2an);

\draw (x1an.east) to node [anchor= south, above] {$h_{11}$} (y1an.west);
\draw (x1an.east) to  (y2an.west);
\draw (x2an.east) to node [anchor= north, below] {$h_{KK}$} (y2an.west);
\draw (x2an.east) to  (y1an.west);

\node at (3,2.2) {$h_{1K}$};
\node at (3,0.9) {$h_{K1}$};

\node[rectangle,draw, inner sep=0.2cm, text width =0.2cm] (r) at (2.5,5) {$R$};

\node[circle,draw,inner sep=0,text width =0.2cm] (ran1) at (2,5.2) {};
\node[circle,draw,inner sep=0,text width =0.2cm] (rank) at (2,4.8) {};
\draw[thick] (2.2, 5.2) -- (2.1, 5.2);
\draw[thick] (2.2, 4.8) -- (2.1, 4.8);
\node[circle,draw,inner sep=0,text width =0.05cm]  at (2.1, 5.05) {};
\node[circle,draw,inner sep=0,text width =0.05cm]  at (2.1, 4.95) {};
\node (r_an_l) at (2,5) {};

\node[circle,draw,inner sep=0,text width =0.2cm] (ran1) at (3,5.2) {};
\node[circle,draw,inner sep=0,text width =0.2cm] (rank) at (3,4.8) {};
\draw[thick] (2.8, 5.2) -- (2.9, 5.2);
\draw[thick] (2.8, 4.8) -- (2.9, 4.8);
\node[circle,draw,inner sep=0,text width =0.05cm]  at (2.9, 5.05) {};
\node[circle,draw,inner sep=0,text width =0.05cm]  at (2.9, 4.95) {};
\node (r_an_r) at (3,5) {};

\draw[very thick] (x1an.east) to (r_an_l.west);
\draw[very thick] (x2an.east) to (r_an_l.west);
\draw[very thick] (r_an_r.east) to (y1an.west);
\draw[very thick] (r_an_r.east) to (y2an.west);

\node[very thick] at (4.5,2) {$\mathbf{g}_{Kr}$};
\node[very thick] at (0.6,2) {$\mathbf{g}_{rK}$};
\node[very thick] at (4.5,4) {$\mathbf{g}_{1r}$};
\node[very thick] at (0.6,4) {$\mathbf{g}_{r1}$};

\node[circle, inner sep=0, text width=0.2cm] at (0,-0.8) {};
\end{tikzpicture}
}
}
\caption{The wireless relay-assisted network with layer one repeaters and one smart relay is shown in subfigure (a). The dotted lines demonstrate the equivalent links between source the destination taking into account of the presence of the repeaters. All paths from source to destination nodes take two time slots and links from source to relay and relay to destination take one time slot.  The equivalent channel is established in subbfigure (b) by replacing the relay as an instantaneous relay and information through instantaneous relay arrive at destinations the same time as the direct links. \label{fig:irc}}
\end{center}
\vspace{-1cm}
\end{figure}

In a SISO IRC, an example of two sources and two destinations shown in Fig. \ref{fig:irc}, we denote the sources as $S_i$ and destinations as $D_i$, $i=1,\ldots,K$. The multi-antenna relay node
is denoted as $R$. Denote the complex channel from $S_i$ to $D_j$ as $h_{ji}$ and the complex channel vector from $S_i$ to $R$ as $\g_{ri}$ and from $R$ to $D_j$ as $\g_{jr}$. All channels are assumed to be independent identically distributed complex Gaussian variates, $\g_{ri}, \g_{jr} \in \mathbb{C}^{M \times 1}$, where $M$ is the number of antennas at the relay. We assume a memoryless instantaneous relay \cite{ElGamal2005, Lee2011,Chang2011} which has a linear processing matrix $\R\in \mathbb{C}^{M \times M}$. The signal received at $R$ is:
\begin{equation}
 \mathbf{y}_r= \sum_{j=1}^K \g_{rj} x_j + \mathbf{n}_r
\end{equation}
where $x_i$  are the transmit symbols from $S_i$  which is assumed to be zero mean proper Gaussian variable and has \textcolor{black}{power constraint $P_s^{max}$}, $\mathbb{E} \left[|x_j|^2 \right]=P_j \leq P_s^{max}, j=1, \ldots, K$. The noise at the relay is denoted as $\n_r$ which is assumed to be independent identically distributed proper Gaussian variables with zero mean and unit variance. The assumption of circularity for the transmit symbols simply the derivation as the achievable rate is the Shannon rate. For the degrees-of-freedom and achievable rates improvement with improper Gaussian signaling, please refer to \cite{Cadambe2010, Ho2011a} respectively. The received signal at $D_j,j=1,\ldots,K$, is 
\begin{equation}\label{eqt:in_out}
 y_j= \sum_{l=1}^K \left( h_{jl} + \g_{jr}^H \R \g_{rl}\right) x_l + \g_{jr}^H \R \mathbf{n}_r + n_j
\end{equation}
For brevity, denote $\p=[P_1,\ldots, P_K]^T \in \mathbb{R}_+^{K \times 1}$. The Signal-to-Interference-and-Noise ratio at destination $j$ is 
\begin{equation}\label{eqt:sinr_no_in}
 \sinr_j(\R, \p)= \frac{|h_{jj} + \g_{jr}^H \R \g_{rj}|^2 P_j}{\sum_{l=1,l \neq j}^{K}|h_{jl} + \g_{jr}^H \R \g_{rl}|^2  P_l + \|\g_{jr}^H \R \|^2 + 1}
\end{equation}
where $\|\g_{jr}^H \R \|^2$ is the amplified noise from R to $D_j$. The power constraint at the relay is
\begin{equation}
 \mathbb{E}_{\mathbf{y}_r} \left[ \tr\left( \R \mathbf{y}_r \mathbf{y}_r^H \R^H\right) \right] \leq P_r^{max}.
\end{equation} \textcolor{black}{Note that $\mathbb{E}_{\n_r, x_j,j=1,\ldots,K} \left[ \mathbf{y}_r \mathbf{y}_r^H\right]= \sum_{j=1}^K \g_{rj} \g_{rj}^H P_j + \mathbf{I}$. The power constraint is therefore rewritten as the following:
\begin{equation}\label{eqt:pow_con1}
 \tr \left( \R \Q \R^H \right)\leq P_r^{max}
\end{equation} where 
\begin{equation}\label{eqt:Q}
 \Q=\sum_{j=1}^K \g_{rj} \g_{rj}^H P_j + \mathbf{I}.
\end{equation}
}

\subsection{Interference neutralization}

In order to neutralize interference, the following $K(K-1)$ equations have to be satisfied at the same time:
\begin{equation}\label{eqt:in_con}
         h_{ij} + \g_{ir}^H \R \g_{rj} =0, \hspace{1cm} i,j=1,\ldots, K, i \neq j.   
\end{equation}
 Let $\G_{dr}=[\g_{1r}, \g_{2r}, \ldots, \g_{Kr}]$ and $\G_{rt}=[\g_{r1}, \g_{r2}, \ldots, \g_{rK} ]$. Denote $s_l= \g_{lr}^H \R \g_{rl}$. We have the interference neutralization requirement,
\begin{equation}\label{eqt:vecin}
 \G_{dr}^H \R \G_{rt} = \bS
\end{equation}
 with 
\begin{equation}\label{eqt:S}
 \bS= \left[ \begin{array}{ccc}
              s_1 & \ldots & -h_{1K}\\
	      -h_{21} & s_2 & \ldots\\
	       & \ddots & \\
	      \ldots & -h_{K (K-1)} & s_K
             \end{array}
\right].
\end{equation}
\textcolor{black}{
Note that $\bS$ is a matrix with off-diagonal elements as the channel coefficients of the interference channel and diagonal elements as the optimization variables $\s$. As we sill show later, the optimization can be facilitated if the optimization variable is $\bS$ instead of $\s$. This is due to the linear relationship between $\R$ and $\bS$. However, we must stress that only the diagonal elements of $\bS$, $\s$, are free to be optimized as $\bS$ is constrained to the form in \eqref{eqt:S}. Eqt. \eqref{eqt:S} can be rewritten to a more comprehensive form. We introduce a row selection matrix $\T$ which select the off-diagonal elements of $\bS$ from the vector $\bvec(\bS)$. For example when $K=2$, we have $\T=[0,1,0,0; 0,0,1,0]$ and $\T \bvec(\bS)= [-h_{21}, -h_{12}]^T$. We have
\begin{equation}\label{eqt:S2}
 \T \bvec(\bS)= - \T \bvec(\bH).
\end{equation} 
}
If IN is feasible (see later in Theorem \ref{lem:fea} for feasibility issue), we can choose a relay matrix $\R$, which is a function of complex coefficients $\s$, that satisfies the IN constraint \eqref{eqt:vecin} and achieves the following SINR,
\textcolor{black}{
\begin{equation}\label{eqt:sinr_in}
 \sinr_j^{\IN}(\bS,P_j)= \frac{|h_{jj} + \g_{jr}^H \R(\bS) \g_{rj}|^2 P_j}{ \|\g_{jr}^H \R(\bS) \|^2 + 1}
\end{equation} In the following, we use the matrices $\R$ and $\R(\bS)$ interchangeably when we wish to emphasize the optimization parameter $\bS$. Please see the next lemma for the formulation of $\R(\bS)$.}
In order to make sure the requirements in  \eqref{eqt:vecin} are not over-determined, we find the minimum number of antennas required in the relay such that  \eqref{eqt:vecin} is feasible.
\textcolor{black}{
\begin{Lemma}\label{lem:Rvec}
A sufficient condition of IN in \eqref{eqt:vecin} on the minimum number of antennas at the relay is $M \geq K$. For some target signal coefficients  $s_1, \ldots, s_K \in \mathbb{C}$, the relay processing matrix $\R$ that satisfies the interference neutralization  requirement \eqref{eqt:vecin} is determined by
\begin{equation}
 \bvec(\R)= \left( \G_{rt}^{T} \otimes \G_{dr}^{H} \right)^{\dagger} \bvec(\bS).
\end{equation}
\end{Lemma}
\begin{proof}
From \eqref{eqt:vecin}, using $\bvec(\A\B\C)= (\C^T \otimes \A)\bvec(\A)$ \cite{Lutkepohl1996}, we obtain $\left( \G_{rt}^{T} \otimes \G_{dr}^{H} \right) \bvec(\R)= \bvec(\bS)$. For a consistent system \cite[P.12, Section 0.4.2]{Horn1985}, a sufficient condition is to have the number of equations (length of $\bvec(\bS)$) less than or equal to the number of variables (length of $\bvec(\R)$) $M\geq K$. Multiply both sides with the Moore-Penrose inverse of $\left( \G_{rt}^{T} \otimes \G_{dr}^{H} \right)$ and result follows.
\end{proof}
Note that when $M < K$, the system in \eqref{eqt:vecin} has more equations than unknowns and the consistency of the system depends heavily on the particular channel realizations.
For simplicity of presentation, from now on, we set $M=K$. We assume that the matrix $\left( \G_{rt}^{T} \otimes \G_{dr}^{H} \right)$ is full rank and obtain \cite[P.35]{Lutkepohl1996}
\begin{equation} \label{eqt:relay_Rs}
 \bvec(\R)= \left( \G_{rt}^{T} \otimes \G_{dr}^{H} \right)^{-1} \bvec(\bS).
\end{equation}
} If there is no power constraint at the relay, given any target signal coefficients $s_i$, we can construct a relay matrix $\R$ as in \eqref{eqt:relay_Rs} such that the $\IN$ requirement is satisfied. Otherwise, any performance metrics, e.g. sum rate, subject to the IN and power constraints, are optimized in a dimension of $K$ by optimizing the complex coefficients $\s=[s_1,\ldots, s_K]^T$.

\section{The Pareto Boundary optimization problem}\label{sec:pb_formulation}
\textcolor{black}{The achievable rate region of an instantaneous AF relay is defined to be the set of rate tuple achieved by all possible relay processing matrix $\R$ satisfying the power constraint \eqref{eqt:pow_con1}:
\begin{equation}\label{eqt:ach_rate_reg}
 \mathcal{R}= \bigcup_{\R:\tr\left( \R \Q \R^H\right) \leq P_r} \left( C(\sinr_1(\R)), \ldots, C(\sinr_K(\R))\right)
\end{equation}
where $C(x)=\log_2(1+x)$. Similarly, we define the achievable rate region of an instantaneous AF relay with IN to be the set of rate tuple achieved by all possible relay processing matrix $\R$ satisfying the IN constraint \eqref{eqt:relay_Rs} and the power constraint \eqref{eqt:pow_con1}:
\begin{equation}\label{eqt:ach_rate_reg_in}
 \mathcal{R}^{\IN}= \bigcup_{\substack{\R:\tr\left( \R \Q \R^H\right) \leq P_r,\\ 
\bvec(\R)= \left( \G_{rt}^{T} \otimes \G_{dr}^{H} \right)^{-1} \bvec(\bS),\\
 \T \bvec(\bS)= - \T \bvec(\bH)}} 
\left( C(\sinr_1(\R)), \ldots, C(\sinr_K(\R))\right).
\end{equation}
Note that the IN requirement gives a smaller feasible set and thus $\mathcal{R}^{IN} \subset \mathcal{R}$. However, the motivation of the study of $\mathcal{R}^{\IN}$ is two-fold: (a) intuitively when both transmit power at S and R is very large, the optimal relay strategy is to neutralize interference, so as to transmit at the maximum DOF; but the performance of other SNR regimes is yet to be studied. (b) the characterization of $\mathcal{R}^{\IN}$ in the instantaneous IRC with direct link between S-D pairs is still an open problem. The outer boundary of $\mathcal{R}$ $(\mathcal{R}^{\IN})$ -- the Pareto boundary (PB) of $\mathcal{R}$ $(\mathcal{R}^{\IN})$ -- is
a set of operating points at which one user cannot increase its own rate without simultaneously decreasing other users rate.
\begin{Def}[\cite{Jorswieck2008,Ho2010a}]
 A rate-tuple $(r_1,\ldots, r_K)$ is Pareto optimal if there is no other rate tuple $(q_1,\ldots, q_K)$ such that $(q_1,\ldots,q_K)\geq(r_1,\ldots, r_K)$ and 
$(q_1,\ldots,q_K)\neq(r_1,\ldots, r_K)$\footnote{The inequality is component-wise.}. 
\end{Def}
By definition, the operation points on the PB can be evaluated by maximizing one user's rate while keeping other users' rates constant. Other optimization techniques for PB evaluation
has been proposed \cite{Charafeddine2007,Zhang2010a}. Here, the PB problem is formulated as the maximization of $\sinr_1$ subject to the constraints on $\sinr_j \geq \gamma_j$ for some pre-determined target SINR values $\gamma_j$, $j=2,\ldots,K$.}
\textcolor{black}{
\begin{Prob}\label{prob_pb}
The Pareto boundary of $\mathcal{R}$ \eqref{eqt:ach_rate_reg} is a set of rate tuple $\left( C(\gamma_1^{\#}), C(\gamma_2), \ldots, C(\gamma_K) \right)$ where $\gamma_1^{\#}$ is the optimal objective value and $\gamma_j, j=2,\ldots K$ are the constraints of the following optimization problem.
\begin{equation}\label{def:PB_no_in}
 (\PB)\left\{
 \begin{aligned}
  \max_{\R \in \mathbb{C}^{M}, \p\in \mathbb{R}_+^{K \times 1}} \hspace{0.3cm} & \sinr_1(\R,\p)\\
  \mbox{s.t. } \hspace{0.3 cm}& \sinr_j(\R,\p) \geq \gamma_j, \hspace{0.3cm} j=2,\ldots, K,\\
  \hspace{0.3cm} & \tr \left( \R \Q\R^H \right)\leq P_r^{max}.
 \end{aligned} \right.
\end{equation}
\end{Prob}
}
\textcolor{black}{
Similarly, we formulate the PB optimization problem with IN in the following. To this end, we combine the first two constraints in the feasibility set of \eqref{eqt:ach_rate_reg_in}.
 Using the fact that $\tr \left( \A\B\C\D \right)= \bvec(\D^T)^T \left( \C^T \otimes \A\right) \bvec(\B)$ \cite{Lutkepohl1996}, we have
\begin{equation}
 \begin{aligned}
  \tr (\R \Q \R^H) &= \bvec(\R)^H \left( \Q^T \otimes \I\right) \bvec(\R)\\
& \overset{(a)}{=} \bvec(\bS)^H \underbrace{\left( \left( \diag(\p) + \G_{rt}^{-*} \G_{rt}^{-T} \right) \otimes \G_{dr}^{-1} \G_{dr}^{-H} \right)}_{\tilde{\Q}} \bvec(\bS)
 \end{aligned}
\end{equation} 
where (a) is due to \eqref{eqt:Q}, \eqref{eqt:relay_Rs} and $\G_{rt}^{-*}\Q^T \G_{rt}^{-T}= \diag(\p) + \G_{rt}^{-*} \G_{rt}^{-T}$. 
}
%
\textcolor{black}{
\begin{Prob} 
The Pareto boundary of $\mathcal{R}^{\IN}$ \eqref{eqt:ach_rate_reg_in} is a set of rate tuple $\left( C(\gamma_1^{\#}), C(\gamma_2), \ldots, C(\gamma_K) \right)$ where $\gamma_1^{\#}$ is the optimal objective value and $\gamma_j, j=2,\ldots K$ are the constraints of the following optimization problem.
\begin{tikzpicture}
  \node at (-0.5,-0.1) {$(\PB-\IN)$};
  \draw[decorate,decoration={brace}] (1,-2) -- (1,1.5);
  \node at (6.0,0) {
    \begin{minipage}{0.7\linewidth}
      \begin{subequations}\label{def:PB_in}
	\begin{align}
	\max_{\bS \in \mathbb{C}^{K}, \p\in \mathbb{R}_+^{K \times 1}} \hspace{0.3cm} & \sinr_1^{\IN}(\bS,P_1)\\
	\mbox{s.t. }\hspace{0.3cm} & \sinr_j^{\IN}(\bS,P_j) \geq \gamma_j, \hspace{0.3cm} j=2,\ldots, K,\\
	  & \bvec(\bS)^H \tilde{\Q} \bvec(\bS) \leq P_r^{max}, \label{eqt:def_PB_in_pow}\\
	&  \T \bvec(\bS)= - \T \bvec(\bH). \label{eqt:def_PB_in_in}
	\end{align} 
      \end{subequations}              
    \end{minipage}
  };
\end{tikzpicture}
\end{Prob}
} 
%
Before we show the optimization methods of the aforementioned problems, the feasibility issue of \eqref{def:PB_in} needs to be addressed.
\textcolor{black}{
\begin{Theorem}\label{lem:fea}
A necessary and sufficient condition for the feasibility of interference neutralization - satisfying \eqref{eqt:def_PB_in_pow} and \eqref{eqt:def_PB_in_in} simultaneously - is
\begin{equation}\label{eqt:fea_cond}
  \left(\T \bvec(\bH) \right)^H \left( \T \tilde{\Q}^{-1} \T^H \right)^{-1}  \T \bvec(\bH) \leq P_r^{max},
\end{equation}  where $\tilde{\Q}= \left(\diag(\p)+ \G_{rt}^{-*} \G_{rt}^{-T} \right) \otimes \left(\G_{dr}^{-1} \G_{dr}^{-H} \right)$.
A feasible solution is
\begin{equation}\label{eqt:fea_sol}
 \bvec(\bS)=\F \left(\x_n + (\T \F)^\dagger \T \bvec(\bH)\right)
\end{equation} where $\x_n \in \mathcal{N}(\T \F)$ and $\tilde{\Q}^{-1}=\F\F^H$. 
\end{Theorem}
}
\begin{proof}
See Appendix \ref{app:fea}.
\end{proof}
Note that the left hand side of the condition \eqref{eqt:fea_cond} is the power of the relay matrix which only performs IN (with $\x_n=\0$ in \eqref{eqt:fea_sol}) and only depends on channel coefficients whereas the right hand side is the relay power constraint. The importance of Theorem \ref{lem:fea} lies in the practicality of IN feasibility verification. With the information of channel qualities $\mathbf{H},\G_{rt},\G_{dr}$, transmit power at S $\p$ and transmit power constraint at relay $P_r^{max}$, we can immediately check whether IN is feasible or not. Further, if IN is feasible, there is always one feasible solution in \eqref{eqt:fea_sol} by setting $\x_n= \0_{K \times 1}$; if there is excess power at relay, we can optimize $\x_n$, and in turn $\bvec(\bS)$, to improve system performance, which is the objective of the succeeding section.

\section{The optimal relay strategies with uninformed S-D nodes}\label{sec:opt_relay}
In this section, we analyze the PB problems taking into consideration that the S and D nodes are uninformed of the presence of the relay nodes and therefore do not optimize their transmit power values: \textcolor{black}{ $\p=\p_0$ for some pre-determined power values $\p_0$. If the S-D pairs are non-cooperative, then each S transmits at full power and thus $\p= \1 P_s^{max}$. } Given the transmit power values of the source nodes, we maximize the achievable rate by choosing the relax matrix. The $\PB$ problems are formulated into quadratically constrained quadratic programs (QCQP) and are then relaxed to semi-definite programs (SDP). The relaxed problems are convex and can be solved using efficient convex optimization tools such as CVX \cite{Grant2011}. We show that in some scenarios, the optimal solution of the relaxed problems  attains the optimality of the original problems. In other words, the convex optimization methods solve the original problem efficiently in such scenarios (Please refer to Corollary \ref{Cor:rank1_no_in} and Lemma \ref{thm:rank1} for details). \textcolor{black}{The procedures used to obtain the PB are summarized in Algorithm \ref{algo}.}

 \begin{algorithm}
\textcolor{black}{
\caption{The pseudo-code for the Pareto boundary optimization in \eqref{def:PB_no_in} and \eqref{def:PB_in} . \label{algo}}
\begin{algorithmic}[1]
  \For{$j=2 \to K$}
  \State Compute the single-user-point of user $j$: $\gamma_j^{max}= \max_{\R} \frac{|h_{jj} + \g_{jr}^H \R \g_{rj}|^2 P_j}{\| \g_{jr}^H \R\|^2 +1}$, in which the user $j$ is the only user in the system with no interference.
  \State For a predefined integer $N$, let $\bbV_j=\left\{ 0, \frac{\gamma_j^{max}}{N-1}, \frac{2\gamma_j^{max}}{N-1},\ldots, \gamma_j^{max} \right\}$.
  \EndFor
  \State Define a tuple $\s$ to be a vector of possible values of $\gamma_j$, $\s=[\gamma_2,\ldots,\gamma_K]$ and $\s \in \bbS=\bbV_2 \times \bbV_3 \times \cdots \times \bbV_K$.
  \For {each $\s \in \bbS$}  
  \State With input parameter $\s$, matrix $\bS$ is defined as in \eqref{eqt:S}.  Solve the optimization problems  \eqref{eqt:convex_no_in_relaxed} for general relay processing matrix optimization or \eqref{eqt:convex_pb} for relay processing matrix optimization with interference neutralization.
  \If {the optimization problem is feasible}
   \State the optimal value $\gamma_1^\dagger$ and $\s$ form a point on the Pareto boundary, which is a $K-1$ dimension hyper-surface, in the $K$ dimensional space.
  \Else 
    \State the values of $\s$ are unachievable subject to the constraints.  
  \EndIf
  \EndFor
\end{algorithmic}
}
 \end{algorithm}

\subsection{The Pareto boundary: general relay optimization}
By introducing an auxiliary variable, we formulate the PB optimization problem with general relay processing matrix as a QCQP. The optimization variable is of dimension $M^2+1$ as compared to $M^2$ number of elements in the relay matrix. Nevertheless, this provides a more structural formulation and amends the analysis as shown in the sequel.

\begin{Lemma}\label{lem:convex_no_in}
\textcolor{black}{
The Pareto boundary of an IRC with instantaneous AF relay \eqref{def:PB_no_in} is a rate tuple $(C(\gamma_1^{\#}), C(\gamma_2),\ldots, C(\gamma_K))$ in which $\gamma_1^{\#}= \frac{\bar{\bv}^H \X_{11} \bar{\bv}}{\bar{\bv}^H \X_{12} \bar{\bv}}$ and $\bar{\bv}$ is the optimal solution of the following optimization problem. The values $\gamma_j,j=2,\ldots, K$ contribute to the constraints of the optimization problem.}
\begin{equation}\label{eqt:convex_no_in}
\left\{
 \begin{aligned}
  \max_{\bv \in \mathbb{C}^{(M^2+1) \times 1}} \hspace{0.3cm} & \frac{\bv^H \X_{11} \bv}{\bv^H \X_{12} \bv}\\
\mbox{s.t.} \hspace{0.3cm} & \frac{\bv^H \X_{j1} \bv}{\bv^H \X_{j2} \bv} \geq \gamma_j, \hspace{0.3cm} j=2,\ldots,K,\\
& \bv^H \X_3 \bv \leq 0.
 \end{aligned} \right.
\end{equation}
The matrices, for $i=1,\ldots,K$, are given by
\begin{equation*}
 \begin{aligned}
  \X_{i1}&=\left[\begin{array}{cc}
              \left( \g_{ri} \otimes \g_{ir}^*\right)^*\left( \g_{ri} \otimes \g_{ir}^*\right)^T & \left( \g_{ri} \otimes \g_{ir}^*\right)^T h_{ii}\\
h_{ii}^*\left( \g_{ri} \otimes \g_{ir}^*\right)^T & |h_{ii}|^2
             \end{array} \right] P_i,\\
\X_{i2} &=\sum_{l\neq i}^{K} \left[\begin{array}{cc}
              \left( \g_{rl} \otimes \g_{ir}^*\right)^*\left( \g_{rl} \otimes \g_{ir}^*\right)^T & \left( \g_{rl} \otimes \g_{ir}^*\right)^T h_{il}\\
h_{il}^*\left( \g_{rl} \otimes \g_{ir}^*\right)^T & |h_{il}|^2
             \end{array} \right] P_l
 + \left[\begin{array}{cc}
                    \I_{M} \otimes (\g_{ir} \g_{ir}^H) & \0_{M^2 \times 1}\\
\0_{1\times M^2}& 1
                   \end{array} \right],\\
\X_{3} &= \left[\begin{array}{cc}
                 \Q^T \otimes \I & \0_{M^2 \times 1}\\
\0_{1\times M^2}& - P_r
                \end{array}
 \right].
 \end{aligned}
\end{equation*}
\end{Lemma}
\begin{proof}
 See Appendix \ref{app:convex_no_in}.
\end{proof}

Define $\V=\bv\bv^H$, we obtain the following convex problem after removing the constraint $\rank(\V)=1$:
\begin{equation}\label{eqt:convex_no_in_relaxed}
\left\{
 \begin{aligned}
  \max_{\V \in \mathbb{C}^{(M^2+1)}, \V \succeq 0 } \hspace{0.3cm} & \tr\left( \X_{11} \V \right)\\
\mbox{s.t.} \hspace{0.3cm} & \tr\left(\X_{12} \V \right)=1,\\
& \tr \left( \left( \X_{j1} -  \gamma_j \X_{j2} \right) \V \right)\geq 0, \hspace{0.3cm} j=2,\ldots,K,\\
& \tr \left( \X_3 \V \right) \leq 0
 \end{aligned} \right.
\end{equation}
which is a semi-definite program (SDP) as matrices $\X_{i1},\X_{i2},\X_3$ are Hermitian and can be solved efficiently using SDP solvers.
\begin{Corollary}\label{Cor:rank1_no_in}
\textcolor{black}{ By \cite[Theorem 3.2]{Huang2010} the rank of the optimum solution of \eqref{eqt:convex_no_in_relaxed} is smaller than $\sqrt{K+1}$. 
In the scenario of $K=2$ S-D pairs, the rank of the optimal solution is one which means that the relaxation is tight and the obtained solution is the global optimal solution of \eqref{eqt:convex_no_in}. In occasions when the optimal solution returned by CVX is not rank-one, one can find a vector which satisfies the same constraint and objective values in \eqref{eqt:convex_no_in_relaxed} using the rank-one reduction procedures \cite[Theorem 2.3]{Ai2009}. Such vector is thus one of the global optimal solutions of \eqref{eqt:convex_no_in}. For more S-D pairs, the result of SDP in \eqref{eqt:convex_no_in_relaxed} is not rank-one. However, one can apply randomization approximation techniques \cite{Luo2010} to approximate the optimal solution of \eqref{eqt:convex_no_in}.}
\end{Corollary}

\subsection{The Pareto boundary with IN}
We manipulate the problem \eqref{def:PB_in} in the same fashion as in \eqref{eqt:convex_no_in}. The optimization problem \eqref{def:PB_in} has one more constraint (the IN constraint) and is thus optimized in a smaller dimension $K^2+1$ rather than $M^2+1$. \textcolor{black}{Note that the condition given in Theorem \ref{lem:fea} is a necessary and sufficient condition for the feasibility of IN (non-empty feasible set in \eqref{eqt:ach_rate_reg_in}). In other words, if the condition is not satisfied, then the optimization problem in \eqref{eqt:pb_pro1} is not feasible regardless of the target SINR values $\gamma_2, \ldots, \gamma_K$.}

\begin{Lemma}\label{lem:pb}
 \textcolor{black}{
The Pareto boundary of an IRC with instantaneous AF relay and IN \eqref{def:PB_in} is a rate tuple $(C(\gamma_1^{\#}), C(\gamma_2),\ldots, C(\gamma_K))$ in which $\gamma_1^{\#}= \bar{\y}^H \hat{\B}_{i1} \bar{\y}$ and $\bar{\y}$ is the optimal solution of the following optimization problem. The values $\gamma_j,j=2,\ldots, K$ contribute to the constraints of the optimization problem.}
\begin{equation}\label{eqt:pb_pro1}
\left\{ \begin{aligned}
 \max_{\y \in \mathbb{C}^{(K^2+1) \times 1} } \hspace{1cm} & \y^H \hat{\B}_{11} \y \\
\mbox{s.t.} \hspace{1cm} & \y^H \hat{\B}_{12} \y=1,\\
& \y^H \hat{\B}_j \y \geq 0, \hspace{1cm} j=2,\ldots, K,\\
& \y^H \hat{\D}_3 \y  \leq 0,\\
& \y^H \hat{\D}_4 \y =0. 
 \end{aligned} \right.
\end{equation} where
\begin{equation*}
 \begin{aligned}
  \hat{\B}_{11}&=  \left[\begin{array}{cc}
                       \bL^T \e_1 \e_1^T \bL & \bL^T \e_1 \h_{11}\\
			\h_{11}^* \e_1^T \bL & |h_{11}|^2
                      \end{array}
 \right] P_1,\\
\hat{\B}_{12}&= \left[\begin{array}{cc}
                       \G_{rt}^{-*} \G_{rt}^{-T} \otimes \e_1 \e_1^T 
  & \0_{K^2 \times 1}\\
			\0_{1 \times K^2} & 1
                      \end{array}
 \right],\\
\hat{\B}_{j}&=  \left[\begin{array}{cc}
                        \bL^T \e_j \e_j^T \bL - \gamma_j \left(\G_{rt}^{-*} \G_{rt}^{-T} \otimes \e_j \e_j^T \right)  & \bL^T \e_j \h_{jj}\\
			\h_{jj}^* \e_j^T \bL & |h_{jj}|^2-\gamma_j
                      \end{array}
 \right], \\
\hat{\D}_{3}&=  \left[\begin{array}{cc}
                       \tilde{\Q} & \0_{K^2 \times 1}\\
			\0_{1 \times K^2} & -P_r^{max}
                      \end{array}
 \right] ,\\
\hat{\D}_{4}&= \left[\begin{array}{cc}
                       \T^H \T & \T^H \T \bvec(\mathbf{H})\\
			\bvec(\mathbf{H})^H \T^H \T & \bvec(\mathbf{H}) \T^H\T \bvec(\mathbf{H})
                      \end{array}
 \right].
 \end{aligned}
\end{equation*}
\end{Lemma} \vspace{0.3cm}
\begin{proof}
 See Appendix \ref{app:pb}.
\end{proof}

Let $\Y=\y\y^H$. Using SDR techniques, the problem in \eqref{eqt:pb_pro1} can be relaxed to the following problem by dropping the rank one constraint on $Y$.
Problem \eqref{eqt:convex_pb} is a convex problem, in particular a semi-definite program, which can be solved efficiently.
\begin{equation}\label{eqt:convex_pb}
\left\{ \begin{aligned}
 \max_{\Y \in \mathbb{C}^{K^2+1}, \Y \succeq 0} \hspace{1cm} & \tr \left(\hat{\B}_{11} \Y \right)\\
\mbox{s.t.} \hspace{1cm} & \tr \left(\hat{\B}_{12} \Y\right)=1,\\
&\tr\left( \hat{\B}_j \Y\right) \geq 0, \hspace{1cm} j=2,\ldots, K,\\
& \tr\left(\hat{\D}_3 \Y \right) \leq 0,\\
& \tr\left( \hat{\D}_4 \Y\right) =0. 
 \end{aligned} \right.
\end{equation} 
Note that, if the optimal solution in \eqref{eqt:convex_pb} is rank one, then such solution solves \eqref{eqt:pb_pro1} optimally. If the optimal solution in \eqref{eqt:convex_pb} is not rank one, then the optimality of \eqref{eqt:pb_pro1} can no longer be guaranteed. In the following, we characterize the rank of the optimal solution of \eqref{eqt:convex_pb}.

\begin{Lemma}\label{thm:rank1}
 The optimal solution of $K$-user IRC Pareto boundary problem with IN in \eqref{eqt:convex_pb}, $\tilde{\Y}$, satisfies
\begin{equation}
 \rank(\tilde{\Y})\leq \sqrt{K+1}.
\end{equation}
\end{Lemma}
\begin{proof}
\textcolor{black}{
Note that the matrix $\hat{\D}_4 \succeq 0$ in  \eqref{eqt:convex_pb} with rank $K^2-K$ and therefore we have eigenvalue decomposition of $\hat{\D}_4$ as
\begin{equation}
 \hat{\D}_4= [\V, \V_0] \diag(\lambda_1,\ldots,\lambda_{K^2-K},\0_{1 \times \left((K^2+1)- (K^2-K)\right)}) \left[\V^H; \V_0^H \right]
\end{equation}
where $\V$ is the eigenvector matrix of $\hat{\D}_4$ corresponding to the eigenvalues $\lambda_1 \geq \ldots \geq \lambda_{K^2-K} \geq 0$; $\V_0$ therefore spans the null space of $\hat{\D}_4$.
Let $\y=\V_0 \x$ for $\x \in \mathbb{C}^{(K^2+1) \times 1}$ and rewrite the optimization problem \eqref{eqt:convex_pb} to
\begin{equation}\label{eqt:qcqp1}
\left\{\begin{aligned}
  \max_{\x \mathbb{C}^{(K^2+1) \times 1}} \hspace{0.3cm} & \x^H \V_0^H \hat{\B}_{11}\V_0 \x\\
\mbox{s.t.} \hspace{0.3cm} & \x^H \V_0^H\hat{\B}_{12}\V_0 \x=1,\\
& \x^H \V_0^H\left(\hat{\B}_{j1}- \gamma_j \hat{\B}_{j2} \right) \V_0\x \geq 0, \hspace{0.3cm} j=2,\ldots, K,\\
& \x^H \V_0 \hat{\D}_3 \V_0 \x \leq 0.
 \end{aligned} \right. 
\end{equation}
There are now $K+1$ constraints in \eqref{eqt:qcqp1} and result follows from \cite[Theorem 3.2]{Huang2010}.}
\end{proof}
\begin{Corollary}\label{cor:rank}
\textcolor{black}{ The optimization methods and results for \eqref{eqt:convex_pb} are similar to Corollary \ref{Cor:rank1_no_in} and shall not be repeated here.}
\end{Corollary}

\section{The optimal transmit power at informed simple S-D nodes}\label{sec:opt_power}
In the scenario where the source and destination nodes are informed of the presence of the relay and are willing to improve the rate performance of the system by cooperation, we can improve further the Pareto boundary by optimizing the transmit power at the source nodes. In the $K$-user SISO-IC, the power allocation at the $\PB$ is characterized and is obtained by searching over $K-1$ real-valued parameters \cite{Charafeddine2007}. 
\textcolor{black}{In the following, we observe that given the relay matrix $\R$, the Pareto optimal source power in IC does not apply to the IRC. The is due to the dependence of the source power in the relay power constraint. Nevertheless, we obtain the Pareto optimal source power as a function of relay matrix $\R$ and target SINR values $\gamma_2, \ldots, \gamma_K$. 
\begin{Theorem}\label{thm:opt_pow_no_in}
 For any given relay matrix $\R$, the $\PB$ of \eqref{def:PB_no_in} is a set of rate tuple $(C(\gamma_1^{\#}), C(\gamma_2),\ldots, C(\gamma_K))$ which are attained by 
the optimal transmit power at S nodes, $\p^{\#}=[p_1^{\#},\p_{2:K}^{\#}]\geq \0$:
\begin{equation}\label{eqt:opt_pow_no_in}
\left\{
\begin{aligned}
 p_1^{\#} &=  \min \left( P_s^{max}, \frac{ P_r- \tr(\R\R^H)-  \bc^T \left[\A \right]_{(:,2:K)}^{-1} \q}{\left(\| \g_{r1}\|^2 - \bc^T\left[\A \right]_{(:,2:K)}^{-1}\ba_1 \right)}, \min_k \left( \frac{ P_s^{max} \sum_{j=2}^K [\A]_{kj} - [\q]_k}{[\A]_{k1}} \right)\right) \\
 \p^{\#}_{2:K} &= \min \left(\left[\A\right]_{(:,2:K)}^{-1} \left(\q- \ba_1 p_1^{\#}\right), P_s^{max} \right) 
\end{aligned} \right.
\end{equation}
where $\bc^T= \left[\|\g_{r2}\|^2,\ldots,\| \g_{rK}\|^2 \right]$, $\left[\q\right]_m =  \gamma_{m+1} \left(  \|\g_{(m+1)r}^H \R \|^2 + 1 \right)$ and 
\begin{equation}
  [\A]_{ml} =\left\{ \begin{array}{cc}
                                    |h_{(m+1)(m+1)} + \g_{(m+1)r}^H \R \g_{r(m+1)}|^2 & \mbox{if } m+1=l,\\
				     - \gamma_{m+1} |h_{(m+1)l} + \g_{(m+1)r}^H \R \g_{rl}|^2 & \mbox{if } (m+1) \neq l.
                                   \end{array}
\right. , 1 \leq m \leq K-1.
\end{equation} and $[\A]_{(2:K)}$ is a matrix which consists of the second to $K$-th columns of the matrix $\A$.
\end{Theorem}
}
\vspace{0.1cm}
\begin{proof}
 See Appendix \ref{app:opt_pow_no_in}.
\end{proof}
Theorem \ref{thm:opt_pow_no_in} gives the Pareto optimal transmit power $\p^{\#}$ as a function of the given relay matrix $\R$ and the target SINR values $\gamma_2,\ldots, \gamma_K$. Similar results can be obtained when IN is employed, as shown below.
\textcolor{black}{
\begin{Theorem}\label{thm:pow_in}
 For any given relay matrix $\R$, the $\PB$ of \eqref{def:PB_in} is a set of rate tuple $(C(\gamma_1^{\#}), C(\gamma_2),\ldots, C(\gamma_K))$ which are attained by 
the optimal transmit power at S nodes, $\bu=[u_1,\ldots,u_{K}]\geq \0$:
\begin{equation}\label{eqt:pow_in}
\left\{
\begin{aligned}
 u_1 &= \min \left( \frac{P_r^{max}-\tr \left( \R(\s) \R(\s)^H \right)}{ \|\g_{r1} \R(\s) \|^2 }  - \sum_{j=2}^K \frac{\|\g_{rj} \R(\s)\|^2}{\|\g_{r1}\R(\s) \|^2} u_j, P_s^{max} \right)\\
 u_j  & = \min \left( \frac{\gamma_j \left(\|\g_{jr}^H \R(\s) \|^2 + 1 \right)}{ |h_{jj} + \g_{jr}^H \R(\s) \g_{rj}|^2 }, P_s^{max} \right), \hspace{0.3cm} j=2,\ldots, K.
\end{aligned} \right.
\end{equation}
\end{Theorem}
}
\vspace{0.1cm}
\begin{proof}
 See Appendix \ref{app:pow_in}.
\end{proof}

\textcolor{black}{Note that Theorems \ref{thm:opt_pow_no_in} and \ref{thm:pow_in} give the Pareto optimal source power allocation for a given relay processing matrix $R$. While the joint optimization of the source power and relay processing matrix is highly complicated, one can approach the problem by solving iteratively (a) the relay processing matrix given the source power \eqref{eqt:convex_no_in_relaxed} and \eqref{eqt:convex_pb} and (b) the source power given the relay processing matrix \eqref{eqt:opt_pow_no_in} and \eqref{eqt:pow_in}. However, the iterative optimization approach may not converge to the global optimal solution.}

In the following section, we provide numerical evidence of performance gain of a relay introduction to a SISO interference channel. In particular, in a setting of uninformed source nodes, we show the rate improvement of solely introducing and optimizing the relay strategy whereas in a setting of informed S-D nodes, we compare the rate performance of the general relay optimization and the relay optimization with IN to the rate region of a SISO IC.

\section{Simulation Results}\label{sec:simulations}
 \textcolor{black}{For illustrative purposes, we let $K=M=2$. We assume that each element in the channel matrices $\mathbf{H}, \G_{rt}, \G_{dr}$ is an independent identically distributed complex Gaussian variable with zero mean and unit variance. 
In Section \ref{sec:sim_uninform}, we simulate the achievable rates of the SISO IC (marked as squares). For comparison, we simulate the achievable rates of the fully connected IRC with the same S-D nodes, by introducing a relay, equipped with 2 antennas, to the aforementioned IC, and the relay can choose to enforce IN (marked with asterisks) or not (marked with triangles). We show that optimized relay strategies improve achievable rate regions. In Section \ref{sim:sumrate} and \ref{sim:fairrate}, we compare the average sum rate and proportional fairness utility achieved by optimized relay strategies and by power allocation on the IC respectively. In Section \ref{sim:power}, we illustrate the sum rate performance and proportional fairness utility when the transmit power constraint at source nodes and relay nodes vary. }

\textcolor{black}{
\subsection{Rate region improvement}\label{sec:sim_uninform}
In Fig. \ref{fig:rate_improve_1}, we plot the achievable rate region of a two-user SISO IC with transmit power constraint at each source node $P_s^{max}=10dB$. Introducing an instantaneous relay, equipped with 2 antennas, we obtain an IRC. We set the relay power constraint as $P_r^{max}=20dB$. The achievable rates achieved by general relay optimization and IN outperform the IC case. The black arrows originate from the Nash Equilibrium point: the rate points in which both users transmit with full power. The north-east side of the arrows mark the rate region improved by the relay, in the scenario of uninformed source nodes. This validates our intuition that optimized relay strategies can improve achievable rates of the system even if the source nodes are oblivious to the existence of the relay and do not change their transmit power. Further note that the single user points achieved in IRC with general relay optimization always outperform the single user points in a SISO IC. It demonstrates that the relay not only is capable of reducing interference in the system but also forwarding the desired signal to the destinations. }

\textcolor{black}{
In Fig. \ref{fig:rate_improve_2}, we reduce the relay transmit power to $P_r^{max}=10dB$. We observe that the rate region achieved by IN reduce significantly because the relay is not able to neutralize interference and improve desired signal quality with limited power. 
}
%
\begin{figure}[!ht]
\begin{center}
\includegraphics[width=12cm, height=9cm,keepaspectratio]{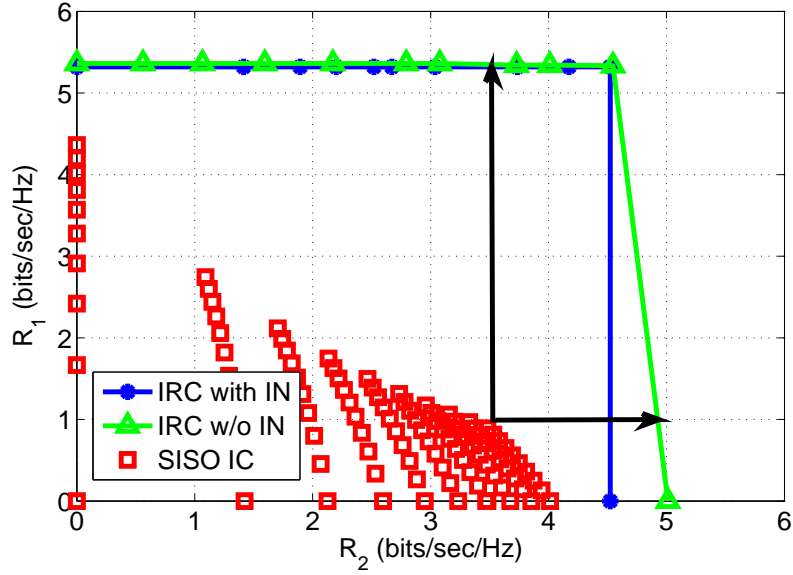}
\caption{The rate improvement of relay optimization on a two-user SISO IRC with $K=M=2$, $P_r^{max}=20dB$, $P_s^{max}=10dB$.  The arrow marks the increment of rate region by introducing a relay into the system and optimizing the relay strategy.}
\label{fig:rate_improve_1}
\end{center}
\end{figure}

\begin{figure}[!ht]
\begin{center}
 \includegraphics[width=12cm, height=9cm,keepaspectratio]{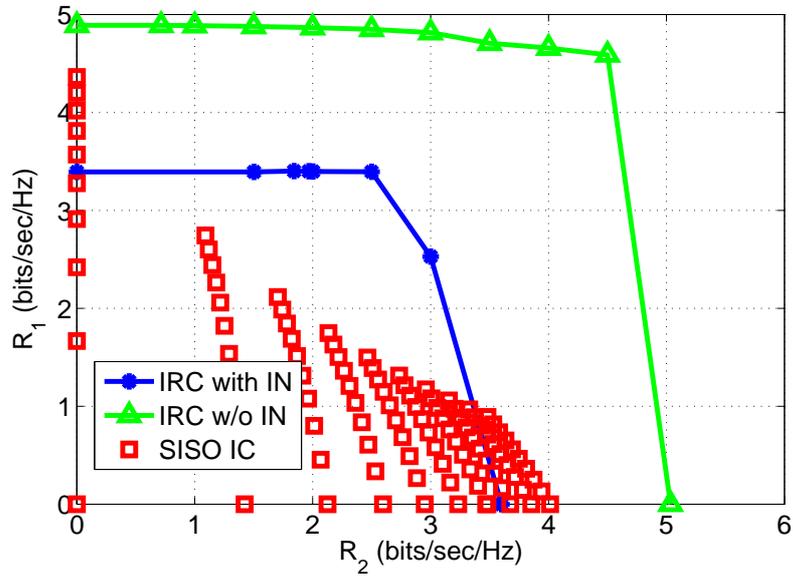}
\caption{The rate improvement of relay optimization on a two-user SISO IRC with $K=M=2$, $P_r^{max}=10dB$, $P_s^{max}=10dB$.}
\label{fig:rate_improve_2}
\end{center}
\end{figure}

\textcolor{black}{
\subsection{Average sum rate improvement}\label{sim:sumrate}
In Fig. \ref{fig:sumrate}, we show the maximum sum rate achieved by general relay optimization, IN and power allocation on the IC, averaged over 100 independent channel realizations. The power constraint at the source node is assumed to be $P_s^{max}=10dB$ and we increase the relay transmit power from $5dB$ to $25dB$. We observe that the optimized relay strategy without IN always outperform the maximum sum rate of the IC, demonstrating that an instantaneous smart relay can improve average sum rate performance. Further, we observe that the performance of IN is limited by the relay transmit power. Although IN is analytically appealing, there are limitations of the implementation of IN. Such scenarios include strong interference channels in which the receivers have strong interference from other transmitters in the system. In this case, more power at the relay may be required to completely null out interference and if such power is not available to the relay, then IN is not feasible. On the other hand, if the strength of the interference channel is not strong, enforcing IN, the relay loses its optimization degrees of freedom and may not be able to achieve some operating points as the general relay optimization would achieve.
}
\begin{figure}[!ht]
\begin{center}
 \includegraphics[width=12cm, height=9cm,keepaspectratio]{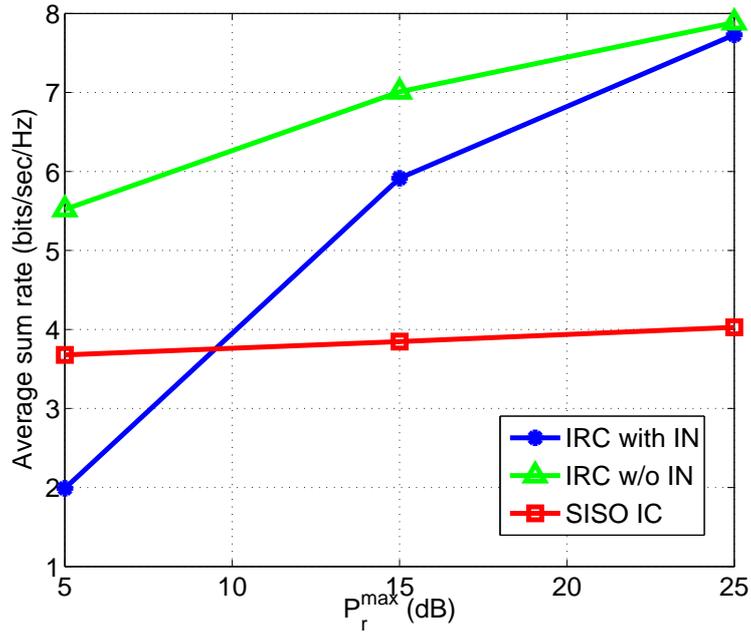}
\caption{The average sum rate a two-user SISO IRC with $K=M=2$, $P_s^{max}=10dB$. The optimized relay strategies improve the average sum rate of the system significantly.}
\label{fig:sumrate}
\end{center}
\end{figure}

\textcolor{black}{
\subsection{Proportional fairness improvement}\label{sim:fairrate}
While average sum rate is an important system performance measure, user fairness holds importance in many applications. In Fig. \ref{fig:fairrate}, we illustrate the average proportional fairness utility which is defined as the $\max_{R_1,R_2}(R_1-R_1^{NE})(R_2-R_2^{max})$. We observe that the optimized relay strategies, with and without IN, provide promising proportional fairness and better sum rate performance compared to IC.
}
\begin{figure}[!ht]
\begin{center}
 \includegraphics[width=12cm, height=9cm,keepaspectratio]{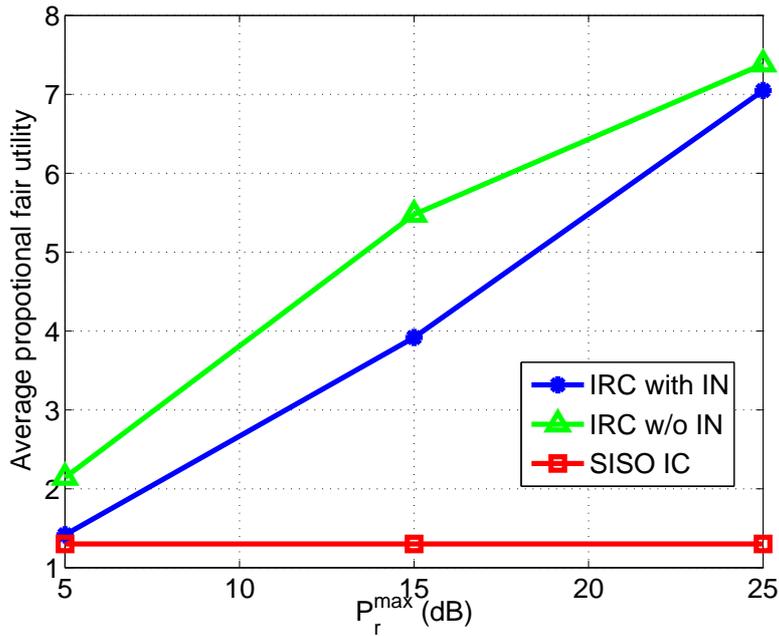}
\caption{The average proportional fairness utility $(R_1-R_1^{NE})(R_2-R_2^{NE})$ of a two-user SISO IRC with $K=M=2$, $P_s^{max}=10dB$. The optimized relay strategies improve fairness of the system significantly.}
\label{fig:fairrate}
\end{center}
\end{figure}

\textcolor{black}{
\subsection{Performance measures in terms of transmit power constraints}\label{sim:power}
It is interesting to observe that the performance of optimized relay strategies depend on both $P_s^{max}$ and $P_r^{max}$. This is due to the amplify-and-forward nature of the relay. If the transmit power from source nodes is high, the relay can spend less power on amplification of signals due to the relay power constraint. In Fig. \ref{fig:sumrate_in}, we plot the maximum sum rate achieved by optimized relay strategy with and without IN and the maximum sum rate in IC. Note that the feasibility conditions of IN, shown in Theorem \ref{lem:fea}, is validated in Fig. \ref{fig:sumrate_in}. For a fixed relay power $P_r^{max}$, when the source power increases such that the conditions are violated, IN is not feasible and the achievable rate is zero. For the general relay optimization, for a fixed relay power and increasing source power, the rate performance is not always increasing because the increased source power increases the interference power and the relay may not have enough power to manage interference and amplify desired signals simultaneously. 
}

\textcolor{black}{
In Fig. \ref{fig:fairrate_in}, we show the maximum proportional fairness utility achieved by optimized relay strategies and IC. When both source power and relay power are abundant, the fairness is desirable. However, when the source power (which is also the strength of interference) is relatively stronger than the relay power, the fairness achieved by the relay strategies is overtaken by the proportional fairness utility achieved by IC.
}
\begin{figure}[!ht]
\begin{center}
\includegraphics[width=12cm, height=10cm,keepaspectratio]{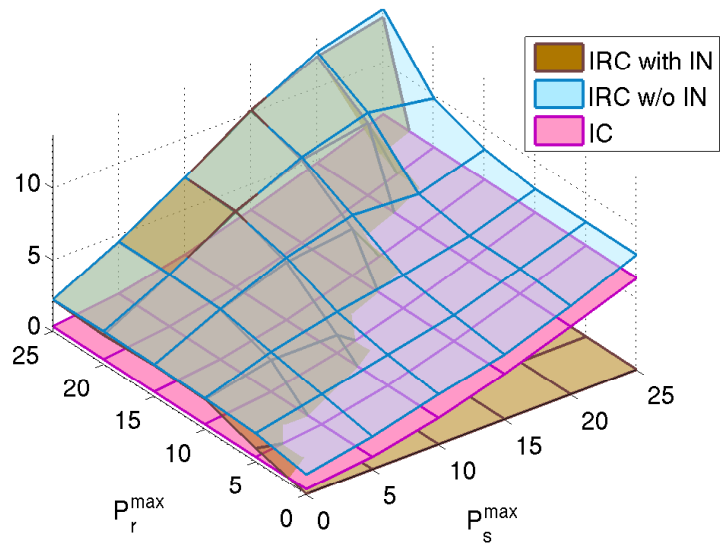}
\caption{The sumrate of a particular channel realization of a two-user SISO IRC with $K=M=2$. }
\label{fig:sumrate_in}
\end{center}
\end{figure}

\begin{figure}[!ht]
\begin{center}
 \includegraphics[width=12cm, height=10cm,keepaspectratio]{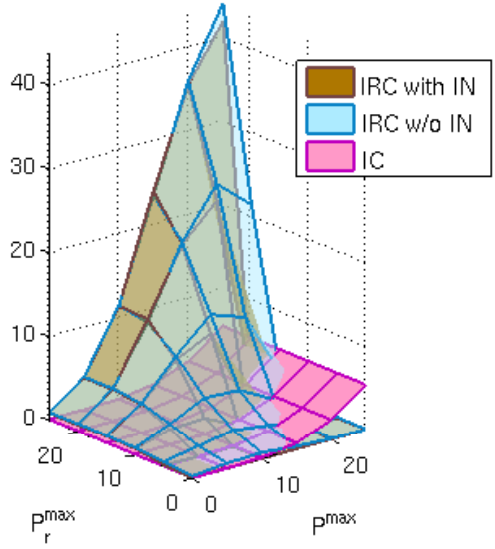}
\caption{The proportional fairness utility $(R_1-R_1^{NE})(R_2-R_2^{NE})$ of a particular channel realization of a two-user SISO IRC with $K=M=2$.}
\label{fig:fairrate_in}
\end{center}
\end{figure}

\section{Conclusion and Future research directions}\label{sec:conclusion}
The achievable rate region of a SISO-IC has been an on-going research topic, with recent interest on the question whether a relay introduction to the SISO-IC, obtaining an interference relay channel, provides any performance gain. In this paper, we study this problem by assuming an instantaneous  amplify-and-forward relay with uninformed source and destination nodes in the system. We examine the gain of rate region of the relay introduction by formulating the Pareto boundary problem with optimization over relay processing matrix. The optimization problems, with and without the employment of interference neutralization techniques, are solved using semi-definite relaxation techniques. The global optimality of the solutions are proved in the scenario of two source and two destination nodes. In the scenario of informed source nodes, we allow the source nodes to optimize their transmit power. The transmit power values at the source nodes which attain the Pareto boundary are obtained in closed-form.  Simulation results confirm that instantaneous relay is able to improve the achievable rate region, even in the scenario of uninformed source nodes; improve average sum rate and average proportional fairness of the system. 

This paper motivates the study of performance of the IRC with an AF relay which can be implemented easily in practical applications as the relay is only responsible for a simple forward process which does not incur a processing delay as compared to other complicated relays. As a preliminary study, we only allow the relay to choose between IN or no IN. To evaluate the full potential  of the AF relay, one may allow the relay to choose different relay strategies, e.g. interference forwarding, signal amplifying, interference neutralization, etc.,  depending on its power budget and the channel qualities in the system. Another interesting problem is the physical placement of the relay with the goal of rate performance improvement.

\section{Appendix}

\subsection{Proof of Theorem\ref{lem:fea}}\label{app:fea}
\textcolor{black}{
The goal of this section is to prove that the necessary and sufficient condition of $\bS$, satisfying
\begin{equation*}
\left\{
 \begin{aligned}
  & \bvec(\bS)^H \tilde{\Q} \bvec(\bS) \leq P_r^{max}\\
  & \T \bvec(\bS)= -\T \bvec(\bH)
 \end{aligned} \right.
\end{equation*} with $\tilde{\Q}= \left(\diag(\p)+ \G_{rt}^{-*} \G_{rt}^{-T} \right) \otimes \left(\G_{dr}^{-1} \G_{dr}^{-H} \right)$, is 
\begin{equation}
  \left(\T \bvec(\bH) \right)^H \left( \T \tilde{\Q}^{-1} \T^H \right)^{-1}  \T \bvec(\bH) \leq P_r^{max}.
\end{equation}
Perform eigenvalue-value decomposition on the Hermitian matrix $\tilde{\Q}=\mathbf{U}^{H} \vGamma \mathbf{U}$ and we let 
$\x= \vGamma^{1/2} \mathbf{U} \bvec(\bS)$ and $\x^H \x\leq P_r^{max}$.
Let $\F=\mathbf{U}^{H} \vGamma^{-1/2}$ and we have
\begin{equation} \label{eqt:t_h}
 \T \bvec(\bS) = \T \F \x = -\T \bvec(\bH).
\end{equation} 
Note that the matrix $\T \F$ is of dimension $K(K-1) \times K^2$ and has nullity of $K$. Denote the null space of $\T \F$ by $\mathcal{N}(\T \F)$ and we choose
\begin{equation}\label{eqt:x_char}
\x=\x_n + \x_h, \hspace{1cm} \x_n \in \mathcal{N}(\T \F), \; \x_h \notin \mathcal{N}(\T \F).
\end{equation}
From \eqref{eqt:t_h} and \eqref{eqt:x_char}, we have
\begin{equation}\label{eqt:in}
 \x_h= -\left(\T\F \right)^\dagger \T \bvec(\bH).
\end{equation} The sufficient condition of IN is thus
\begin{equation}\label{eqt:pow_con}
  \|\x_h\|^2 = \|(\T \F)^\dagger \T \bvec(\bH) \|^2= (\T \bvec(\bH))^H \left( \T \F (\T \F)^H \right)^{-1} \T \bvec(\bH) \leq P_r^{max}.
\end{equation}
Once $\x_h$ is chosen according to  \eqref{eqt:in} and satisfy \eqref{eqt:pow_con}, we are free to choose $x_n\in \mathcal{N}(\T\F)$  as long as $\|\x\|^2 \leq P_r^{max}$. }

\textcolor{black}{
To see that \eqref{eqt:pow_con} is a necessary condition, we need to prove that if IN is feasible then \eqref{eqt:pow_con} must be true. We prove by contradiction and assume that IN is feasible and it is possible to find a solution $\x$ such that $\|\x\|^2 \leq P_r^{max}$ but $(\T \bvec(\bH))^H \left( \T \F (\T \F)^H \right)^{-1} \T \bvec(\bH) > P_r^{max}$. By assumption, there exist $\x$ such that $-\T \bvec(\bH)=\T \F \x$. We multiply both sides by $(\T \bvec(\bH))^H \left( \T \F (\T \F)^H \right)^{-1}$ and we have
\begin{equation}\label{eqt:step1a}
 (\T \bvec(\bH))^H \left( \T \F (\T \F)^H \right)^{-1} \T \bvec(\bH)  \overset{(a)}{=} -(\T \bvec(\bH) )^H \left( \T \F (\T \F)^H \right)^{-1} \T \F \x > P_r^{max}.
\end{equation} Since the product on the left of the equality $(a)$ is a real number, the product on the right side of $(a)$ must be a real number also. Thus, we can write
\begin{equation}\label{eqt:step2a}
 \x= -a (\T \F)^H \left( \T \F (\T \F)^H\right)^{-1} \T \bvec(\bH), \hspace{1cm}, a \in \mathbb{R}^{+}.
\end{equation} Substitute  \eqref{eqt:step2a} into  \eqref{eqt:step1a}, we have
\begin{equation}
 a (\T \bvec(\bH))^H \left( \T \F (\T \F)^H \right)^{-1} \T \bvec(\bH) > P_r^{max}.
\end{equation} By assumption, $(\T \bvec(\bH))^H \left( \T \F (\T \F)^H \right)^{-1} \T \bvec(\bH) > P_r^{max}$ and thus $a>1$. Computing the norm of $\x$, we have $\|\x\|^2= a^2 (\T \bvec(\bH))^H (\T \F (\T \F)^H)^{-1} \T \bvec(\bH)> P_r^{max}$ which violates the power constraint. Thus, $(\T \bvec(\bH))^H \left(\T \F (\T \F)^H\right)^{-1} \T \bvec(\bH) \leq P_r^{max} $ is a necessary and sufficient condition for interference neutralization.}

\subsection{Proof of Lemma \ref{lem:convex_no_in}}\label{app:convex_no_in}
\textcolor{black}{
In this section, we show that steps to obtain \eqref{eqt:convex_no_in} from \eqref{prob_pb}.
From  \eqref{eqt:sinr_no_in}, the desired signal power of $D_i$ 
\begin{equation}
  |h_{ii} + \g_{ir}^H \R \g_{ri}|^2 P_i = |h_{ii} + \left(\g_{ri} \otimes \g_{ir}^* \right)^T \bvec(\R)|^2 P_i
\end{equation}
which is due to $\ba^T \B \bc= \left( \bc \otimes \ba \right)^T \bvec(\B)$.
The interference and noise power of Tx-Rx pair $i$ is then
\begin{equation}
 \begin{aligned}
  & \sum_{l \neq i}^K |h_{il} + \left(\g_{rl} \otimes \g_{ir}^* \right)^T \bvec(\R)|^2 P_l + \|\g_{ir}^H \R \|^2 +1\\
&\overset{(a)}{=} \sum_{l \neq i}^K |h_{il} + \left(\g_{rl} \otimes \g_{ir}^* \right)^T \bvec(\R)|^2 P_l + \|\bvec(\R)^T \Z^T \left(\g_{ir}^* \otimes \I_{K} \right) \|^2 +1
 \end{aligned}
\end{equation}
where the manipulation (a) is due to $\g_{ir}^H \R= \bvec(\R^T)^T(\g_{ir}^* \otimes \I_{K})= \bvec(\R)^T\Z^T(\g_{ir}^* \otimes \I_{K})$ and the matrix $\Z$ which satisfies $\bvec(\R^T)=\Z \bvec(\R) $ is called a commutation matrix \cite[Section 9.2]{Lutkepohl1996}. $\Z$ has elements zeros and ones and $\Z=\Z^T$. Utilizing the Charnes-Cooper Transform \cite{Chalise2009, Zhang2009, Liao2011}, we let $\veta/t=\bvec(\R)$ where $t \in \mathbb{C}$ and $\bv=[\veta^T, t]^T \in \mathbb{C}^{(K^2+1) \times 1}$. The SINR of user $i$ is
\begin{equation}\label{eqt:sinr_no_in_app}
 \begin{aligned}
  \sinr_i &= \frac{|h_{ii}t + \left(\g_{ri} \otimes \g_{ir}^* \right)^T \veta|^2 P_i}{\sum_{l \neq i}^K |h_{il}t + \left(\g_{rl} \otimes \g_{ir}^* \right)^T \veta|^2 P_l + \|\veta^T \Z^T \left(\g_{ir}^* \otimes \I_{K} \right) \|^2 +|t|^2}\\
&\overset{(a)}{=} \frac{|h_{ii}t + \left(\g_{ri} \otimes \g_{ir}^* \right)^T \veta|^2 P_i}{\sum_{l \neq i}^K |h_{il}t + \left(\g_{rl} \otimes \g_{ir}^* \right)^T \veta|^2 P_l + \veta^H \left(\I_{K}  \otimes \g_{ir} \g_{ir}^H\right) \veta +|t|^2}\\
&= \frac{\bv^H \X_{i1} \bv}{\bv^H \X_{i2} \bv}
 \end{aligned}
\end{equation}
where (a) is due to  $\Z \left(\A \otimes \B \right) \Z= \B \otimes \A$ and 
\begin{equation}
\begin{aligned}
       \X_{i1}&=\left[\begin{array}{cc}
              \left( \g_{ri} \otimes \g_{ir}^*\right)^*\left( \g_{ri} \otimes \g_{ir}^*\right)^T & \left( \g_{ri} \otimes \g_{ir}^*\right)^T h_{ii}\\
h_{ii}^*\left( \g_{ri} \otimes \g_{ir}^*\right)^T & |h_{ii}|^2
             \end{array} \right] P_i,\\
\X_{i2} &=\sum_{l\neq i}^{K} \left[\begin{array}{cc}
              \left( \g_{rl} \otimes \g_{ir}^*\right)^*\left( \g_{rl} \otimes \g_{ir}^*\right)^T & \left( \g_{rl} \otimes \g_{ir}^*\right)^T h_{il}\\
h_{il}^*\left( \g_{rl} \otimes \g_{ir}^*\right)^T & |h_{il}|^2
             \end{array} \right] P_l + \left[\begin{array}{cc}
                    \I_{M} \otimes (\g_{ir} \g_{ir}^H) & \0_{M^2 \times 1}\\
		    \0_{1\times M^2}& 1
                   \end{array} \right].
\end{aligned}
\end{equation}
The power constraint at the relay is
\begin{equation}\label{eqt:pow_con_app}
\begin{aligned}
  \tr \left( \R \Q \R^H\right) =\bvec(\R)^H \left( \Q^T \otimes \I \right)\bvec(\R) & \leq P_r^{max}\\
 \Leftrightarrow \bv^H \left[ \begin{array}{cc}
                                                             \left( \Q^T \otimes \I_{M} \right) & \0_{M^2 \times 1}\\
							  \0_{1 \times M^2} & -P_r^{max}
                                                            \end{array}
                                                         \right]\bv & \leq 0
\end{aligned}
\end{equation}
From Eqt. \eqref{eqt:sinr_no_in_app} and \eqref{eqt:pow_con_app}, we obtain the formulation in Lemma \ref{lem:convex_no_in}.}

\subsection{Proof of Lemma \ref{lem:pb}}\label{app:pb}
\textcolor{black}{In this section, we show the steps to obtain \eqref{eqt:pb_pro1} from \eqref{def:PB_in}.}
We rewrite the noise power to the following
\begin{equation}\label{eqt:noise_pow}
\begin{aligned}
\g_{ir}^H \R \R^H \g_{ir}&= \g_{ir}^H \left( \G_{dr}^{-H} \bS \G_{rt}^{-1}\right) \left( \G_{dr}^{-H} \bS \G_{rt}^{-1}\right)^H \g_{ir}\\
&= \g_{ir}^H \G_{dr}^{-H} \bS \G_{rt}^{-1} \G_{rt}^{-H} \bS^H \G_{dr}^{-1} \g_{ir}\\
&\overset{(a)}{=} \bvec(\bS^T)^T \left(\G_{dr}^{-*} \g_{ir}^* \otimes \G_{rt}^{-1} \right) \left(\G_{dr}^{-*} \g_{ir}^* \otimes \G_{rt}^{-1} \right)^H \bvec(\bS^T)^*\\
&\overset{(b)}{=} \bvec(\bS)^T \Z \left(\G_{dr}^{-*} \g_{ir}^* \g_{ir}^T\G_{dr}^{-T} \otimes \G_{rt}^{-1}\G_{rt}^{-H} \right)  \Z \bvec(\bS)^*\\
&\overset{(c)}{=} \bvec(\bS)^T \left( \G_{rt}^{-1} \G_{rt}^{-H} \otimes \G_{dr}^{-*} \g_{ir}^*  \g_{ir}^T \G_{dr}^{-T}  \right) \bvec(\bS)^*\\
&\overset{(d)}{=} \bvec(\bS)^H \left( \G_{rt}^{-*} \G_{rt}^{-T} \otimes \G_{dr}^{-1} \g_{ir}  \g_{ir}^H \G_{dr}^{-H}  \right)  \bvec(\bS)\\
&= \bvec(\bS)^H \left( \G_{rt}^{-*} \G_{rt}^{-T} \otimes \e_i \e_i^T  \right)  \bvec(\bS).
\end{aligned}
\end{equation} The transition $(a)$ is due to the Kronecker product properties, $\mathbf{a}^T \mathbf{B}^T \mathbf{C}= \bvec(\mathbf{B})^T \left( \mathbf{a} \otimes \mathbf{C}\right)$ \cite{Brokes2011}. 
The commutation matrix $\Z$ of dimension $K \times K$ satisfies $\Z \bvec(\bS)= \bvec(\bS^T)$, \cite[Section 9.2]{Lutkepohl1996}. 
The transition $(b)$ uses such properties and $(\A_1 \otimes \B_1)(\A_2 \otimes \B_2)= \A_1 \A_2 \otimes \B_1 \B_2$. Then in transition $(c)$, we use the property $\Z \left(\A \otimes \B \right) \Z= \B \otimes \A$.
In transition $(d)$, we uses the fact that the noise energy is a real scalar and a complex conjugation does not affect its value. The last equality is due to that fact that $\g_{ir}$ is the $i$-th column of $\G_{dr}$ and $\G_{dr}^{-1} \g_{dr}= \left[\G_{dr}^{-1}\G_{dr}\right]_{(:,i)}=\e_i$.

 Now, we rewrite the SINR and power constraints as a function of $\bvec(\bS)$.
The signal power at $D_j$ is rewritten as
\begin{equation}\label{eqt:signal_pow}
 |h_{ii}+s_i|^2= |h_{ii}+ \e_i^T \bL \bvec(\bS)|^2 P_i
\end{equation}
where $\bL$ is a row selection matrix, $\s= \bL \bvec(\bS)$.
From \eqref{eqt:noise_pow} and \eqref{eqt:signal_pow}, the SINR  of $D_i$ is
\begin{equation}\label{eqt:sinr_vp}
 \sinr_i= \frac{|h_{ii}+ \e_i^T \bL \bvec(\bS)|^2 P_i}{\bvec(\bS)^H \B_{i2} \bvec(\bS) +1}.
\end{equation} Recall from \eqref{eqt:def_PB_in_pow} and \eqref{eqt:def_PB_in_in} that any feasible solution satisfies 
\textcolor{black}{
\begin{equation}\label{eqt:in_vp}
\left\{
\begin{aligned}
\bvec(\bS(\s))^H \tilde{\Q} \bvec(\bS(\s)) &\leq P_r^{max}\\
\underbrace{ 
         \T }_{\D_4} \bvec(\bS)= -\T \bvec(\mathbf{H}) &= \bb.
\end{aligned} \right.
\end{equation} } However, the second constraint \eqref{eqt:in_vp} creates complication to the optimization problem because of the asymmetric structure of $\D_4$. We here propose an equivalent constraint 
\begin{equation}\label{eqt:in_vp2}
 \| \D_4 \bvec(\bS) - \bb \|^2 = 0.
\end{equation}
From Eqt. \eqref{eqt:sinr_vp} and \eqref{eqt:in_vp} and let $\w=\bvec(\bS)$, \eqref{def:PB_in} is equivalent to the following formulation,
\begin{equation}\label{eqt:pb_pro}
\left\{  \begin{aligned}
  \max_{\w \in \mathbb{C}^{K^2 \times 1}} \hspace{1cm} & \frac{|h_{11}+ \e_1^T \bL \w|^2 P_1}{\w^H \B_{12} \w +1}\\
\mbox{s.t.} \hspace{1cm} & \frac{|h_{jj}+ \e_j^T \bL \w|^2 P_j}{\w^H \B_{j2} \w+1} \geq \gamma_j, \hspace{1cm} j=2,\ldots, K,\\
& \w^H \D_3 \w \leq P_r^{max},\\
& \|\D_4 \w -\bb\|^2 = 0.
 \end{aligned} \right.
\end{equation}

\textcolor{black}{In the following, we convert the optimization problem in  \eqref{eqt:pb_pro} into the standard QCQP. We proceed with the Charnes-Cooper transform \cite{Chalise2009, Zhang2009,Liao2011}} by
substituting the optimization variable $\hat{\w}=\w t$ for some complex scalar $t \neq 0$. The optimization problem is rewritten as
\begin{equation}
\left\{
 \begin{aligned}
 \max_{\hat{\w},t} \hspace{1cm} & \frac{\hat{\w}^H \left[\e_1^T \bL,   h_{11}\right]^H \left[\e_1^T \bL,   h_{11}\right] P_1\hat{\w}}{\hat{\w}^H \B_{12} \hat{\w} + t^2}\\
\mbox{s.t.} \hspace{1cm}& \frac{\hat{\w}^H \left[\e_j^T \bL,   h_{jj}\right]^H \left[\e_j^T \bL,   h_{jj}\right] P_j \hat{\w}}{\hat{\w}^H \B_{j2} \hat{\w} + t^2} \geq \gamma_j, \hspace{1cm} j=2,\ldots, K,\\
& \hat{\w}^H \D_3 \hat{\w}\leq t^2 P_r^{max},\\
& \| \D_4  \hat{\w} -\bb t\|^2 = 0.
 \end{aligned} \right.
\end{equation}
We denote $\B_{i1}=\left[\e_i^T \bL,   h_{ii}\right]^H \left[\e_i^T \bL,   h_{ii}\right] P_i, i=1,\ldots,K$. Without loss of generality, set the denominator of the objective function to one and define a new optimization parameter $\y= [\hat{\w}^T, t]^T$. We obtain \eqref{eqt:pb_pro1}
with \begin{equation}\label{eqt:B}
 \begin{aligned}
  \hat{\B}_{i1}&= \B_{i1}= \left[\begin{array}{cc}
                       \bL^T \e_i \e_i^T \bL & \bL^T \e_i \h_{ii}\\
			\h_{ii}^* \e_i^T \bL & |h_{ii}|^2
                      \end{array}
 \right] P_i,\\
\hat{\B}_{i2}&= \left[\begin{array}{cc}
                       \B_{i2} & \0_{K^2 \times 1}\\
		\0_{1 \times K^2} & 1
                      \end{array}
 \right]=\left[\begin{array}{cc}
                       \G_{rt}^{-*} \G_{rt}^{-T} \otimes \e_i \e_i^T 
  & \0_{K^2 \times 1}\\
			\0_{1 \times K^2} & 1
                      \end{array}
 \right],\\
\hat{\B}_{j}&= \hat{\B}_{j1}- \gamma_j \hat{\B}_{j2}= \left[\begin{array}{cc}
                        \bL^T \e_j \e_j^T \bL - \gamma_j \left(\G_{rt}^{-*} \G_{rt}^{-T} \otimes \e_j \e_j^T \right)  & \bL^T \e_j \h_{jj}\\
			\h_{jj}^* \e_j^T \bL & |h_{jj}|^2-\gamma_j
                      \end{array}
 \right], \\
\hat{\D}_{3}&= \left[\begin{array}{cc}
                       \D_{3} & \0_{K^2 \times 1}\\
			\0_{1 \times K^2} & -\frac{P_r}{P_t+1}
                      \end{array}
 \right]=  \left[\begin{array}{cc}
                       \G_{rt}^{-*} \G_{rt}^{-T} \otimes \G_{dr}^{-1} \G_{dr}^{-H} & \0_{K^2 \times 1}\\
			\0_{1 \times K^2} & -\frac{P_r}{P_t+1}
                      \end{array}
 \right] ,\\
\hat{\D}_{4}&= \left[\begin{array}{cc}
                       \D_4^H \D_4 & -\D_4^H \bb\\
			- \bb^H\B_{4} & \bb^H\bb
                      \end{array}
 \right]=\left[\begin{array}{cc}
                       \T^H \T & \T^H \T \bvec(\mathbf{H})\\
			\bvec(\mathbf{H})^H \T^H \T & \bvec(\mathbf{H}) \T^H\T \bvec(\mathbf{H})
                      \end{array}
 \right].
 \end{aligned}
\end{equation} Note that all above matrices above are Hermitian matrices.


\subsection{Proof of Theorem \ref{thm:opt_pow_no_in}}\label{app:opt_pow_no_in}
For any given relay matrix $\R$, we write the optimization in  \eqref{def:PB_no_in} as
\textcolor{black}{
\begin{subequations}
 \begin{align}
  \max_{\p\in \mathbb{R}_+^{K \times 1}} \hspace{0.3cm} & \frac{|h_{11} + \g_{1r}^H \R \g_{r1}|^2 P_1}{\sum_{l=1,l \neq 1}^{K}|h_{1l} + \g_{1r}^H \R \g_{rl}|^2  P_l + \|\g_{1r}^H \R \|^2 + 1}\\
\mbox{s.t. }\hspace{0.3cm} & \frac{|h_{jj} + \g_{jr}^H \R \g_{rj}|^2 P_j}{\sum_{l=1,l \neq j}^{K}|h_{jl} + \g_{jr}^H \R \g_{rl}|^2  P_l + \|\g_{jr}^H \R \|^2 + 1} \geq \gamma_j, \hspace{0.3cm} j=2,\ldots,K, \label{eqt:sinr}\\
&  \tr \left( \R \left( \sum_{l=1}^K \g_{rl} \g_{rl}^H P_l + \I \right)\R^H\right) \leq P_r^{max},\label{eqt:sum_pow}\\
& P_l \leq P_s^{max}, \hspace{0.3cm} l=1,\ldots,K. \label{eqt:pow_con2}
 \end{align} 
\end{subequations}
}
If the problem is feasible, at optimality \emph{all} the SINR constraints are settled in equality. To see this, denote the optimal power allocation as $\hat{\p}$ and assume that $\sinr_2(\hat{\p}) > \gamma_2$ and $\sinr_j(\hat{\p})=\gamma_j$ for $j=3,\ldots, K$. Note that $\sinr_2(\hat{\p})$ is monotonically increasing with $\hat{P}_2$ whereas $\sinr_1(\hat{\p})$ and $\sinr_j(\hat{\p})$ are monotonically decreasing with $\hat{P}_2$. Thus, the decreased value of $\hat{P}_2$ increases both $\sinr_1$ and $\sinr_j$, $j=3,\ldots,K$. On the other hand, the power constraints \eqref{eqt:sum_pow} and \eqref{eqt:pow_con2} and the $\sinr_2$ constraint remain valid. This contradicts to the assumption that $\hat{P}_2$ attains the optimal point. 

Since all $\sinr$ constraints are active at optimality, we write all $K-1$ constraints \eqref{eqt:sinr} in the following:
\begin{equation}
 \A \hat{\p} = \q
\end{equation}
where for $m=1,\ldots,K-1$, $l=1,\ldots,K$
\begin{equation}
 \begin{aligned}
   \left[ \A \right]_{ml}&= \left\{ \begin{array}{cc}
                                    |h_{(m+1)(m+1)} + \g_{(m+1)r}^H \R \g_{r(m+1)}|^2 & \mbox{if } m+1=l,\\
				     - \gamma_{m+1} |h_{(m+1)l} + \g_{(m+1)r}^H \R \g_{rl}|^2 & \mbox{if } (m+1) \neq l.
                                   \end{array}
\right.\\
\left[\q\right]_m &=  \gamma_{m+1} \left(  \|\g_{(m+1)r}^H \R \|^2 + 1 \right).
 \end{aligned}
\end{equation}
Note that the matrix $\A$ has dimension $(K-1) \times K$ and we denote the $i$-th column of $\A$ as $\ba_i$ and the second to last elements of $\hat{\p}$ as $\hat{\p}_{2:K}$. We have
\begin{equation}\label{eqt:step3}
 \hat{\p}_{2:K} = \left[\ba_2,\ldots,\ba_K \right]^{-1} \left(\q- \ba_1 \hat{P}_1\right).
\end{equation}
For the brevity of notations, we let $\left[\A \right]_{(:,2:K)}=\left[\ba_2,\ldots,\ba_K \right]$. Substitute into the power constraint \eqref{eqt:sum_pow} and denote $\bc^T=\left[\| \g_{r2}\|^2, \ldots, \|\g_{rK}\|^2 \right]$; we have
\begin{equation}\label{eqt:upper_bound1}
\begin{aligned}
 \| \g_{r1}\R\|^2 \hat{P}_1 + \bc^T \left[\A \right]_{(:,2:K)}^{-1} \left(\q- \ba_1 \hat{P}_1\right) &\leq P_r^{max}-\tr \left( \R \R^H \right) \\
\Leftrightarrow  \hat{P}_1 & \leq \frac{ P_r^{max}-\tr \left( \R \R^H \right) -  \bc^T \left[\A \right]_{(:,2:K)}^{-1} \q }{\left(\| \g_{r1} \R\|^2 - \bc^T\left[\A \right]_{(:,2:K)}^{-1}\ba_1 \right)}.
\end{aligned}
\end{equation}
From \eqref{eqt:step3} and \eqref{eqt:pow_con2}, we have
\begin{equation}
 \begin{aligned}
   \hat{\p}_{2:K} &= \left[\A \right]_{(:,2:K)}^{-1} \left(\q- \ba_1 \hat{P}_1\right) \leq P_s^{max} \1_{(K-1) \times 1}\\
\Leftrightarrow \ba_1 \hat{P}_1 &\geq \q- P_s^{max} \left[\A \right]_{(:,2:K)} \1_{(K-1) \times 1}
 \end{aligned}
\end{equation}
Note that $\ba_1=[-\gamma_2 |h_{21}+ \g_{2r}^H \R \g_{r1}|^2, \ldots, -\gamma_K |h_{K1}+ \g_{Kr}^H \R \g_{r1}|^2]^T$ and we have K-1 upper bounds of $P_1$ and for $1 \leq k \leq K-1$:
\begin{equation}\label{eqt:upper_bound2}
 \hat{P}_1\leq \frac{ P_s^{max} \sum_{j=2}^K [\A]_{kj} - [\q]_k}{[\A]_{k1}}.
\end{equation}

On the other hand, the objective function $\sinr_1$ is monotonically increasing with $\hat{P}_1$. To see this mathematically, we write
\begin{equation}\label{eqt:whatever}
 \begin{aligned}
  &\frac{|h_{11} + \g_{1r}^H \R \g_{r1}|^2 \hat{P}_1}{\sum_{l=1,l \neq 1}^{K}|h_{1l} + \g_{1r}^H \R \g_{rl}|^2  \hat{P}_l + \|\g_{1r}^H \R \|^2 + 1}\\
&= \frac{|h_{11} + \g_{1r}^H \R \g_{r1}|^2 \hat{P}_1}{ \bb^T  \hat{\p}_{2:K} + \|\g_{1r}^H \R \|^2 + 1}\\
&=\frac{|h_{11} + \g_{1r}^H \R \g_{r1}|^2 \hat{P}_1}{ \bb^T  \left[\A \right]_{(:,2:K)}^{-1} \left(\q- \ba_1 \hat{P}_1\right) + \|\g_{1r}^H \R \|^2 + 1}\\
&=\frac{|h_{11} + \g_{1r}^H \R \g_{r1}|^2 \hat{P}_1}{- \bb^T \left[\A \right]_{(:,2:K)}^{-1} \ba_1 \hat{P}_1 + \left(\bb^T  \left[\A \right]_{(:,2:K)}^{-1}\q + \|\g_{1r}^H \R \|^2 + 1\right)}\\
 \end{aligned}
\end{equation}
where $\bb^T=\left[ |h_{12} + \g_{1r}^H \R \g_{r2}|^2, \ldots, |h_{1K} + \g_{1r}^H \R \g_{rK}|^2 \right]$. The last equality is of the form $\frac{z_1 \hat{P}_1}{z_2 \hat{P_1}+ z_3}$ which can be manipulated as
\begin{equation}
 \frac{z_1 \hat{P}_1}{z_2 \hat{P_1}+ z_3}= \frac{z_1}{z_2} \frac{z_2 \hat{P}_1 +z_3 - z_3}{z_2\hat{P}_1 + z_3}= \frac{z_1}{z_2}\left(1- \frac{z_3}{z_2\hat{P}_1 + z_3} \right)
\end{equation}
which is indeed monotonically increasing with $\hat{P}_1$. From \eqref{eqt:upper_bound1} and \eqref{eqt:upper_bound2}, we have:
\begin{equation}
 \hat{P}_1 = \min \left( P_s^{max}, \frac{ P_r^{max}-\tr \left( \R \R^H \right) -  \bc^T \left[\A \right]_{(:,2:K)}^{-1} \q}{\left(\| \g_{r1} \R\|^2 - \bc^T\left[\A \right]_{(:,2:K)}^{-1}\ba_1 \right)}, \min_{k=1,\ldots,K} \left( \frac{ P_s^{max} \sum_{j=2}^K [\A]_{kj} - [\q]_k}{[\A]_{k1}} \right)\right).
\end{equation}

\subsection{Proof of Theorem \ref{thm:pow_in}}\label{app:pow_in}
 With the requirement of $\IN$, the interference from Tx $j$ at Rx $i$ is canceled and therefore the function $\sinr_i(\s,P_i)$ is independent to the transmit power from any other transmitters $j\neq i$. From \eqref{def:PB_in}, for any given $\s$, we have
\textcolor{black}{
\begin{subequations}
 \begin{align}
  \max_{\p\in \mathbb{R}_+^{K \times 1}} \hspace{0.3cm} & \frac{|h_{11} + \g_{1r}^H \R(\s) \g_{r1}|^2 P_1}{ \|\g_{1r}^H \R(\s) \|^2 + 1}\\
  \mbox{s.t. }\hspace{0.3cm} & \frac{|h_{jj} + \g_{jr}^H \R(\s) \g_{rj}|^2 P_j}{ \|\g_{jr}^H \R(\s) \|^2 + 1} \geq \gamma_j, \hspace{0.3cm} j=2,\ldots, K,\label{eqt:sinr_constraint}\\
   & \tr \left( \R \left( \sum_{l=1}^K \g_{rl} \g_{rl}^H P_l + \I \right)\R^H\right) \leq P_r^{max}, \label{eqt:pow_con3}\\
& P_l \leq P_s^{max}, \hspace{0.3cm} l=1,\ldots K.
 \end{align}
\end{subequations}} The $\IN$ constraint disappears from the above optimization problem because it is independent to $\p$. 
\textcolor{black}{Denote the optimal power allocation by $\bu=[u_1,\ldots,u_K]$. We can observe that 
SINR constraints \eqref{eqt:sinr_constraint} and power constraint \eqref{eqt:pow_con3} must be active at optimality.
Otherwise, let $\sinr_2>\gamma_2$. The value of $u_2$ can be decreased by a very small amount $\epsilon$ without violating \eqref{eqt:sinr_constraint} and $u_1$ can be increased by $\epsilon$ without violating \eqref{eqt:pow_con3}. This new $u_1$ increases the objective value and leads to contradiction that we are at the optimal point. 
From \eqref{eqt:sinr_constraint} and \eqref{eqt:pow_con3}, we have
\begin{equation}
 \left\{ \begin{aligned}
  u_j  & = \min \left( \frac{\gamma_j \left(\|\g_{jr}^H \R(\s) \|^2 + 1 \right)}{ |h_{jj} + \g_{jr}^H \R(\s) \g_{rj}|^2 }, P_s^{max} \right), \hspace{0.3cm} j=2,\ldots, K,\\
   \sum_{l=1}^K \|\g_{rl} \R \|^2 u_l &= P_r^{max}-\tr \left( \R(\s) \R(\s)^H \right).
 \end{aligned} \right. .
\end{equation}
Therefore, 
\begin{equation}
u_1 = \min \left( \frac{P_r^{max}-\tr \left( \R(\s) \R(\s)^H \right)}{ \|\g_{r1} \R(\s) \|^2 }  - \sum_{j=2}^K \frac{\|\g_{rj} \R(\s)\|^2}{\|\g_{r1}\R(\s) \|^2} u_j, P_s^{max} \right).
\end{equation}
}

\bibliographystyle{IEEEbib}
\bibliography{bib5} 

\end{document}